\documentclass[twoside,twocolumn,10pt]{article}
\usepackage[sc]{mathpazo}
\usepackage[T1]{fontenc} % Use 8-bit encoding that has 256 glyphs
\usepackage{microtype} % Slightly tweak font spacing for aesthetics
\usepackage{bm}
\usepackage{amsmath,amssymb,amsthm}
\usepackage[english]{babel} % Language hyphenation and typographical rules

\usepackage[hmarginratio=1:1,right=1.5cm,top=2cm,bottom=2cm,left=1.5cm,columnsep=20pt]{geometry} % Document margins
\usepackage[hang, small,labelfont=bf,up,textfont=it,up]{caption} % Custom captions under/above floats in tables or figures
\usepackage{booktabs,pgfplots,tikz} 
\pgfplotsset{compat=1.13}% Horizontal rules in tables
\usepackage{algorithm}
\usepackage[noend]{algpseudocode}
\usepackage{booktabs}
\usepackage{filecontents}
\usepackage[export]{adjustbox}

\usepackage[colorlinks=true,linkcolor=magenta,citecolor=cyan]{hyperref}
\usetikzlibrary{patterns}
\allowdisplaybreaks
\newtheorem{theorem}{Theorem}

\algnewcommand\algorithmicinput{\textbf{Input:}}
\algnewcommand\Input{\item[\algorithmicinput]}

\algnewcommand\algorithmicoutput{\textbf{Output:}}
\algnewcommand\Output{\item[\algorithmicoutput]}

\usepackage{enumitem} % Customized lists
\setlist[itemize]{noitemsep} % Make itemize lists more compact

\usepackage{abstract} % Allows abstract customization
\usepackage{titlesec} % Allows customization of titles
\titleformat{\section}[block]{\large\scshape\centering}{\thesection.}{1em}{} % Change the look of the section titles
\titleformat{\subsection}[block]{\large}{\thesubsection.}{1em}{} % Change the look of the section titles
\providecommand{\keywords}[1]{\textbf{\textit{Keywords---}} #1}
\usepackage{fancyhdr} % Headers and footers
\usepackage{titling} % Customizing the title section

\pretitle{\begin{center}\LARGE\bfseries} % Article title formatting
\posttitle{\end{center}} % Article title closing formatting
\title{Probabilistic $N$-$k$ Failure-Identification for Power Systems} % Article title
\author{%
\textsc{Kaarthik Sundar, Carleton Coffrin, Harsha Nagarajan, Russell Bent}\thanks{Corresponding author: {rbent@lanl.gov}} \\[1ex]
\normalsize Center for Nonlinear Studies, Los Alamos National Laboratory, Los Alamos, NM 
}
\date{} % Leave empty to omit a date
\renewcommand{\maketitlehookd}{%
\begin{abstract}
\noindent This paper considers a probabilistic generalization of the $N$-$k$ failure-identification problem in power transmission networks, where the probability of failure of each component in the network is known a priori and the goal of the problem is to find a set of $k$ components that maximizes disruption to the system loads weighted by the probability of simultaneous failure of the $k$ components. The resulting problem is formulated as a bilevel mixed-integer nonlinear program. Convex relaxations, linear approximations, and heuristics are developed to obtain feasible solutions that are close to the optimum. A general cutting-plane algorithm is proposed to solve the convex relaxation and linear approximations of the $N$-$k$ problem. Extensive numerical results corroborate the effectiveness of the proposed algorithms on small-, medium-, and large-scale test instances; the test instances include the IEEE 14-bus system, the IEEE single-area and three-area RTS96 systems, the IEEE 118-bus system, the WECC 240-bus test system, the 1354-bus PEGASE system, and the 2383-bus Polish winter-peak test system. \\

\noindent\keywords{power network resilience, network interdiction, ($N$-$k$) vulnerability, nonlinear optimization, convex relaxation, cutting-plane algorithm, network flow, DC power flow, AC power flow}
\end{abstract}
}

\begin{document}
\maketitle

\section{Introduction} \label{sec:intro}
In modern society, the electric power transmission grid plays a crucial role in sustaining the socioeconomic systems that civilization depends on. While this dependence is widely observed and discussed, large-scale failures (e.g., Superstorm Sandy, Hurricane Matthew, cyber-attacks) remain a threat. These historical events highlight the need for methods that identify sets of components whose failure leads to significant impacts.
One model that is used for these identifications is the $N$-$k$ interdiction or the $N$-$k$ failure-identification problem \cite{Salmeron2009}, which focuses on identifying an $N$-$k$ contingency, \emph{i.e.,} a set of $k$ critical components of the transmission system whose simultaneous or near-simultaneous failure would maximize the disruption, measured in terms of the amount of load shedding in the system. 
%that must occur to prevent a cascading failure of the system. 
This paper considers a probabilistic generalization of the $N$-$k$ problem, which we will refer to as the PNK, where the probability of failure of each component is known a priori and the objective is to find a set of probable $k$ components that causes the maximum disruption to the transmission system under a simultaneous failure. The probability of failure of each component can be generated using fragility models based on exposure to extreme events \cite{Hazus}, among other things.
%The probability of failure of each component can correspond to historical failure rates or can be generated using fragility models based on exposure to extreme events, among other things. 

It can be useful to consider the problem as a bilevel Stackelberg game (see \cite{Brown2006}) with an attacker and a defender where the attacker's and defender's actions are sequential and the attacker has a perfect model of how the defender will optimally operate the transmission system. The objective of the attacker is to identify exactly $k$ components to maximize the impact to the transmission system, measured in terms of the minimum load that must be shed by the defender weighted by the probability of the simultaneous or near-simultaneous failure of the $k$ components. Hence, $N$ and $k$ refer to the total number of components and the number of them that can be interdicted in the system, respectively. The number of possible $N$-$k$ contingencies, even for small values of $k$, makes total enumeration computationally intractable. %We note that an important effect of modeling the PNK this way is that no $N$-$k$ contingency causes a cascading failure in the transmission system. 
%Addressing models of cascades, in this context, remains an open question and a topic of future work.

\subsection{Related work} \label{subsec:lit_work}
The literature contains many variants of the $N$-$k$ problem, and algorithms that can obtain optimal solutions and heuristics have been extensively studied for these variants. For ease of exposition, we will analyze the work done in the literature using the following three categories: (1) deterministic or stochastic variant, (2) linear or convex nonlinear representation of the physics governing the flow of power in electric transmission systems, and (3) heuristic or exact approach to solving the corresponding problem. The physics governing the flow of power through electric transmission networks is nonlinear and is given by Kirchoff's laws; the equations that govern the physics of power flow are referred to as the alternating current (AC) power flow equations. Typically, problems relating to the failure of transmission lines and buses in an electric transmission system use a linear approximation of the AC power flow equations; the equations obtained after linear approximation are called the linearized direct current (DC) power flow equations. 

We will start by reviewing the literature on the deterministic $N$-$k$ interdiction problem that uses the DC power flow equations to develop exact approaches. We remark that the deterministic $N$-$k$ interdiction problem and the deterministic $N$-$k$ failure-identification problem are the the same; hence, we shall use the two names interchangeably. To the best of our knowledge, the first work in the literature to develop mixed-integer programming-based formulations for the deterministic variant of the PNK using the DC power flow model is by Salmeron et al. \cite{Salmeron2004}. The authors in \cite{Salmeron2004} develop a mixed-integer bilevel program to model the deterministic interdiction of lines and buses in the transmission system. A Benders decomposition algorithm, based on the work in \cite{Salmeron2004}, is developed to solve the same problem in \cite{Alvarez2004} on \textit{small} test instances. The authors in \cite{Alvarez2004} also perform AC analysis using PowerWorld \cite{Simulator2003} on the interdiction plan computed using the DC power flow models; their AC analysis involves computing the percentage deviation of the real power flow on the lines provided by the DC power flow model and PowerWorld. The lower-level problem in \cite{Salmeron2004} is replaced by its dual in \cite{Motto2005} and is approximated using KKT conditions in \cite{Arroyo2005}. The first work to systematically develop a decomposition algorithm to solve a deterministic bilevel $N$-$k$ problem based on Benders decomposition on \textit{large} test instances is \cite{Salmeron2009}. The algorithm developed in \cite{Salmeron2009}, an extension of the work in \cite{Salmeron2004}, is tested on a ``U.S. Regional Grid'' with 5000 buses, 5000 lines, and 500 generators. The works in \cite{Salmeron2009, Alvarez2004} measure disruption in terms of long-term power shedding, ignoring short-term shedding resulting from cascading outages immediately after the attack, and use a DC power flow model to compute this long-term power shedding. More recently, \cite{Bienstock2010} develops computationally efficient algorithms to solve a minimum cardinality variant of the $N$-$k$ problem, where the objective is to find a minimum cardinality attack with a throughput less than a prespecified bound; the DC power flow model is used to model the physics of the transmission system. They also formulate and solve a nonlinear continuous version of the problem where the attacker is allowed to change the transmission line parameters to disrupt the system instead of removing transmission lines. The authors in \cite{Delgadillo2010} study the deterministic $N$-$k$ problem where the system operator is allowed to use both load shedding and line switching as defensive operations via a Benders decomposition algorithm. In the context of this literature, our approach makes two key contributions. We analyze the theoretical and computational effects of using better representations of the AC power flow equations, viz., convex nonlinear relaxations, and we consider a probabilistic generalization.

The first work in the literature to use the full AC power flow model to analyze vulnerability in power systems is \cite{Kim2016}; the work develops two continuous optimization models where the attacker is allowed to increase the impedance of transmission lines, similar to \cite{Bienstock2010}. They develop a Frank--Wolfe decomposition algorithm to compute an optimal solution to the problem. Unlike this paper, the authors in \cite{Kim2016} do not model the removal of any component in the transmission network; modeling component removal would require the introduction of binary variables and possibly disjunctions that require different solution techniques. A different line of work in the literature (see \cite{Donde2008,Pinar2010}) focuses on vulnerability assessment in power grids using a nonlinear approximation of the AC power flow equations where the problems are formulated as bilevel mixed-integer programs. These optimization problems either maximize unmet demands or minimize the number of lines to attack. These two approaches were shown to be equivalent in \cite{Arroyo2010}. The authors in \cite{Pinar2010} also develop an algorithm based on graph partitioning to compute the best $N$-$k$ attack based solely on the generation capacities, the demands, and the transmission line limits while ignoring the power flow constraints; they develop solution methodologies that can compute the best $N$-$k$ attack on a transmission system in the western United States.  %on an Opteron 2.2 GHz processor with 4.4 GFlops/sec theoretical peak and 6 Gbytes of physical memory, on networks with over 16000 transmission lines. 

Other heuristic approaches in the literature for the deterministic interdiction problem include \cite{Bier2007}; they present a greedy algorithm that uses a DC power flow model. Although the heuristic approaches produced favorable scaling, the solutions were not optimal. Heuristic approaches based on genetic algorithms can be found in \cite{Arroyo2013}. 

Another related problem of contingency identification has received considerable attention in the literature (see \cite{Enns1982, Eppstein2012, Davis2011, Soltan2016} and references therein). In particular, \cite{Eppstein2012} develops a heuristic approach to identify multiple contingencies that can initiate cascades on large transmission systems, and \cite{Davis2011} uses current injection-based methods and introduces ``line outage distribution factors'' to identify contingencies. This work in the literature focuses on using a variety of ``criticality'' measures, based on the DC power flow model, for the different components in the system that aid in identifying these contingencies. 

A different line of work for general flow networks is based on a generalized network performance measure. This measure is in turn used to rank the nodes and edges of the network (see \cite{Qiang2008}); the downside of using these ranking schemes on electric transmission systems is that they are not validated using power flow models. 
  
As for the literature in stochastic network interdiction problems, the authors in \cite{Cormican1998} provide a thorough introduction to many variations of stochastic network interdiction. In this context, we remark that the PNK is very different from stochastic variants of the $N$-$k$ interdiction problem that compute an $N$-$k$ contingency that maximizes the expected damage given a set of component-outage scenarios. In contrast, the PNK uses the probability of occurrence of an $N$-$k$ contingency to weight the damage incurred by that corresponding contingency; it seeks to identify an $N$-$k$ failure that maximizes the weighted damage. The PNK generalizes the deterministic $N$-$k$ interdiction problem, i.e., when the likelihood of occurrence of every $N$-$k$ contingency is the same, it exactly reduces to the deterministic problem. Stochastic network interdiction has also been studied in the context of identifying locations for installing smuggled nuclear material detectors \cite{Pan2003}. The authors in \cite{Pan2003} formulate the problem of identifying locations for installing smuggled nuclear material detectors as a two-stage stochastic mixed-integer program with recourse, and showed that the problem is NP-Hard. To the best of our knowledge, the stochastic variants of the network interdiction problem, previously addressed in operations research literature, are only for general capacitated flow networks \cite{Cormican1998, Janjarassuk2008} and not for electric transmission systems. Capacitated flow networks can be represented by linear constraints and in some cases have a total unimodularity structure that can be exploited for developing computationally efficient algorithms \cite{Cormen2009}. Furthermore, using the probability of occurrence of an $N$-$k$ failure as weights to the damage incurred by that failure is a natural generalization and extension of failure modeling in transmission systems. Transmission system operators identify a $N$-$k$ failure as a failure with exactly $k$ components failing, even in a probabilistic sense. To the best of our knowledge, this is the first work in the literature to analyze a probabilistic generalization of the network interdiction problem for electric transmission systems. 
 
\subsection{Contributions} \label{subsec:contributions}
The PNK of this paper can be considered a generalization of the work in \cite{Salmeron2009} in two ways:

\begin{enumerate}

\item We consider a probabilistic outage model rather than a deterministic model.

\item We use the AC power flow equations and their convex relaxations instead of the DC power flow approximation. 
\end{enumerate}

The first generalization makes the objective function concave, after suitable transformations, and the second generalization introduces a set of convex-quadratic constraints instead of a set of linear constraints. The main contributions of this paper are as follows. 

\begin{enumerate}

\item We develop the first formulation of the PNK as a mixed-integer, nonlinear program (MINLP) using the AC power flow equations. 

\item We develop one convex relaxation of the PNK that is based on existing convex relaxations of the AC power flow equations. We develop two approximations of the PNK with linear constraints. 

\item We develop a general cutting-plane algorithm for all formulations of the PNK problem and prove its exactness for certain cases. 

\item We perform extensive computational experimentation that demonstrates our approach is scalable to networks with thousands of nodes.

\end{enumerate}

\section{Nomenclature} \label{sec:nomenclature}
This section presents nomenclature and terminology, including those that are well understood and used routinely by the power systems community. An interested reader is referred to \cite{Coffrin2016} for a detailed description of the terminology. 

\noindent \emph{Sets:} \\
$\mathcal N$ - set of buses (nodes) in the network \\
$\mathcal N(i)$ - set of nodes connected to bus $i$ by an edge \\
$\mathcal E$ - set of \emph{from}\footnote{Although a transmission line is modeled as an undirected edge, the power flowing on the line is asymmetric; \emph{from} and \emph{to} edges are used to model this asymmetric power flowing on a line.} edges (lines) in the network \\
$\mathcal E^r$ - set of \emph{to} edges in the network \\
%$\bar{\mathcal E} = \mathcal E \bigcup \mathcal E^r$ \\
$\mathcal S$ - set of $N$-$k$ contingency scenarios \\
$\mathcal {\tilde E}_s$ - set of damaged edges in $N$-$k$ contingency scenario $s$ \\
$\mathcal E_s$ - set of operational \emph{from} edges in the network during $N$-$k$ contingency scenario $s$ \\
$\mathcal E_s^r$ - set of operational \emph{to} edges in the network during $N$-$k$ contingency scenario $s$ \\
%$\bar{\mathcal E}_s = \mathcal E_s \bigcup \mathcal E_s^r$ \\

\noindent \emph{Variables:} \\
%$V_i = v_i\angle \theta_i$ - AC voltage at bus $i$ \\
$V_i^s = v_i^s\angle \theta_i^s$ - AC voltage at bus $i$ during scenario $s \in \mathcal S$ \\
$S^{gs}_i =  p^{gs}_i + \bm i q^{gs}_i$ - AC power generated at bus $i$ during scenario $s \in \mathcal S$\\
$\ell_i^s$ - percent AC power (load) shed at bus $i$ during scenario $s \in \mathcal S$ \\
$S^{s}_{ij} = p^{s}_{ij} + \bm iq^{s}_{ij}$ - AC power flow on line $(i,j)$ during scenario $s \in \mathcal S$\\
$x_{ij}$ - binary interdiction variable for line $(i,j) \in \mathcal E$\\
$z_s$ - a continuous variable representing the minimum load shed for scenario $s\in \mathcal S$ \\

\noindent \emph{Constants:} \\
$\bm S^d_i = \bm p^d_i + \bm i\bm q^d_i$ - AC power demand at bus $i$ \\
$\bm Y_{ij} = \bm g_{ij} + \bm i\bm b_{ij}$ - admittance of line $(i,j)$ \\
%$\bm Z_{ij} = \bm r_{ij} + \bm i\bm x_{ij}$ - Impedance of line $(i,j)$ \\
$(\bm v_i^l, \bm v_i^u)$ - bounds for voltage magnitude at bus $i$ \\
$(\bm p^{gl}_i, \bm p^{gu}_i)$ - bounds for active power generated at bus $i$ \\
$(\bm q^{gl}_i, \bm q^{gu}_i)$ - bounds for reactive power generated at bus $i$ \\
$\bm \theta_{ij}^{\Delta}$ - maximum phase angle difference across line  $(i,j)$ \\
$\bm t_{ij}$ - thermal limit of line $(i,j)$ \\
$\bm{Pr}_{ij}$ - i.i.d. probability of failure of line $(i,j)$ \\
$\bm{Pr}_s = \prod_{(i,j) \in \mathcal {\tilde E}_s} \bm{Pr}_{ij}$ - probability of $N$-$k$ contingency scenario $s$ \\
$\bm k$ - number of damaged/interdicted lines in each $N$-$k$ contingency scenario \\
%$\bm r$ - reference bus \\

\noindent \emph{Other notations:} \\
$\operatorname{Re}(\cdot)$ - real part of a complex number \\
$\operatorname{Im}(\cdot)$ - imaginary part of a complex number \\ 
$(\cdot)^*$ - conjugate of a complex number \\
$\lvert \cdot \rvert$ - magnitude of a complex number \\
$\angle$ - angle of a complex number 

Unless otherwise stated, all the values are in per-unit (pu). A pu system is the expression of system quantities as fractions of a defined base unit quantity; this is extensively used in power system engineering to normalize the magnitude of different system quantities and improve numerical stability. We also note that the probability of failure of a line $(i,j)$ is assumed to be i.i.d. This is a reasonable assumption in the power system context because although extreme events such as hurricanes, earthquakes and geomagnetic disturbances have a lot of structure, that structure is determined via simulation, and is not considered as an uncertainty in typical fragility assessments (see for instance \cite{Hazus}). Once the extreme event is fixed, individual component failures are conditionally independent of one another during the particular event.

\section{A bilevel MINLP formulation} \label{sec:minlp}
In this section, we formulate the PNK as a bilevel MINLP using the nomenclature introduced in the previous section. Although the approach presented in this paper holds for any type of component (generator, transformer, substation, transmission line), for ease of exposition, we assume that only the transmission lines in the system can be disrupted or interdicted; hence, throughout the rest of the article, $N$ and $\bm k$ represent the total number of transmission lines and the number of them that can be interdicted in the system, respectively. The problem is formulated as a bilevel max-min problem. From here on, we will refer to the max and min problem as the \emph{outer} and \emph{inner} problem, respectively. The formulation is as follows:
\begin{flalign}
& (\mathcal F_{\operatorname{ac}}):  \quad \max_{s \in \mathcal S} \quad \bm{Pr}_s \cdot z_s & \label{eq:obj-outer-minlp} %\quad \text{ subject to:} &  %\\
%& \sum_{(i,j) \in \mathcal E} x_{ij} = \bm k, & \label{eq:nk-minlp} \\
%& p_s = \prod_{(i,j) \in \mathcal {\tilde E}_s} \bm p_{ij}, \,\mathcal {\tilde E}_s = \{(i,j) \in \mathcal E: x_{ij} = 1\}, & \label{eq:ps-minlp} \\
%& x_{ij} \in \{0,1\} \quad \forall (i,j) \in \mathcal E, & \label{eq:x-minlp} 
\end{flalign}
where $z_s$ is defined by the following optimization problem for every scenario $s \in \mathcal S$:
\begin{flalign}
& z_s = \min \quad \sum_{i \in {\mathcal N}} \operatorname{Re}(\bm S^d_i) \ell_i^s \quad \text{ subject to:}& \label{eq:obj-inner-minlp} \\
& \bm v_i^l \leqslant \lvert V_i^s \rvert \leqslant \bm v_i^u \quad \forall i\in \mathcal N, & \label{eq:vmag-minlp} \\
& \bm p^{gl}_i \leqslant \operatorname{Re}(S^{gs}_i) \leqslant \bm p^{gu}_i \quad \forall i\in \mathcal N, & \label{eq:pg-minlp} \\
& \bm q^{gl}_i \leqslant \operatorname{Im}(S^{gs}_i) \leqslant \bm q^{gu}_i \quad \forall i\in \mathcal N, & \label{eq:qg-minlp} \\
& \lvert S_{ij}^s \rvert \leqslant \bm t_{ij} \quad \forall (i,j) \in \mathcal E_s \bigcup \mathcal E_s^r, & \label{eq:thermal-minlp} \\
& -\bm \theta_{ij}^{\Delta} \leqslant \angle V_i^sV_j^{s*} \leqslant \bm \theta_{ij}^{\Delta} \quad \forall (i,j) \in \mathcal E_s, & \label{eq:theta-minlp} \\
& 0 \leqslant \ell_i^s \leqslant 1 \quad \forall i \in \mathcal N, & \label{eq:loadshed-minlp} \\
& S_i^{gs} - (1-\ell_i^s)\bm S_i^d = \sum_{(i,j) \in \mathcal E_s \bigcup \mathcal E_s^r} S_{ij}^s \quad \forall i\in \mathcal N, \text{ and} & \label{eq:acopf1} \\
& S_{ij}^s = \bm Y_{ij}^* V_i^s V_i^{s*} - \bm Y_{ij}^* V^s_i V_j^{s*} \quad \forall (i,j) \in \mathcal E_s \bigcup \mathcal E_s^r. & \label{eq:acopf2}
\end{flalign}

 In the above formulation, the outer problem is analogous to the attacker's problem, where the attacker is interested in computing the $N$-$k$ scenario $s$ that maximizes the damage to the system (see Eq. \eqref{eq:obj-outer-minlp}); here, the attacker measures the damage in terms of the minimum load shed incurred by the defender (the system operator) during scenario $s$ weighted by the probability $\bm{Pr}_s$ of scenario $s$. This objective is shown to be an effective severity indicator for an $N$-$k$ scenario via extensive simulations in the dissertation by Nedic \cite{Nedic2003}. Computing the value of $\bm{Pr}_s$ for every $N$-$k$ scenario $s$, presents a mathematical challenge. In the following section, we will show how the independence of the line failure probabilities $\bm{Pr}_{ij}$ can be leveraged to address this mathematical challenge via the use of bilinear terms which are convexified using logarithms.
%  We remark that Eq. \eqref{eq:obj-outer-minlp} assumes that this probability $\bm{Pr}_s$ is known for every $N$-$k$ scenario $s$, which in practice, is not the case. The probability of failure of a line $(i,j) \in \mathcal E$, $\bm{Pr}_{ij}$, is the quantity that is usually known and independence allows one to compute $\bm{Pr}_s$ for any scenario $s$ using the line failure probabilities, $\bm{Pr}_{ij}$, for every line $(i,j) \in \mathcal E$. Modeling this explicit relationship between $\bm{Pr}_s$ and the corresponding $\bm{Pr}_{ij}$ values for every scenario $s \in \mathcal S$ would entail the use of bilinear terms, which is later shown in Section \ref{sec:outer-convex}. } 
 %The constraints \eqref{eq:nk-minlp} and \eqref{eq:ps-minlp} define an $N$-$k$ contingency by enforcing the number of damaged lines in the contingency to be exactly equal to $\bm k$ and the probability of a contingency, respectively. The probability of a contingency is defined as the product of the probabilities of all the damaged lines in that contingency, since the individual line outage probabilities are assumed to be i.i.d. Finally the constraint \eqref{eq:x-minlp} imposes the binary restriction on the interdiction variables $x_{ij}$.

The inner problem is the defender's problem, which is to operate the power network for contingency $s$ in such a way as to minimize the active load shed in the network. Although the defender seeks to minimize the active load shed in the system, constraint \eqref{eq:acopf1} ensures that a certain quantity of reactive load is also shed in the system. This is achieved in constraint \eqref{eq:acopf1} by forcing the load-shedding factor $\ell_i^s$ at each bus $i$ during scenario $s \in \mathcal S$ to act on both the real and reactive loads. This type of load shedding is also referred to as constant power-factor shedding. The constraints for the inner problem can be classified into two categories, viz., the physical limits of the network (operational side constraints) and the AC power flow equations that govern the physics of power flowing in the network. Constraints \eqref{eq:vmag-minlp} and \eqref{eq:theta-minlp} bound the voltage magnitude at each bus and the phase angle difference between the buses connecting a line, respectively. Constraints \eqref{eq:pg-minlp} and \eqref{eq:qg-minlp} limit the real and reactive power generated at a bus during contingency $s$, respectively. Constraint \eqref{eq:loadshed-minlp} gives the restrictions on the load-shedding variable $\ell_i^s$. Eqs. \eqref{eq:acopf1} and \eqref{eq:acopf2} together define how power flows in the network and form the core building block of many power system applications; they are also the source of nonconvexity in the defender's problem. The PNK is a difficult problem (in fact, the inner problem is NP-hard \cite{Verma2009,Lehmann2016}). Hence, in the next section, we present a convex  relaxation of the inner problem and a convexification procedure to handle the nonlinear terms in the outer problem to convert the MINLP to a mixed-integer convex program. We note that the convex relaxations of the AC power flow equations in \eqref{eq:acopf1}, presented in the next section, are established and broadly accepted in the power systems and optimization literature (see \cite{Coffrin2016} and references therein).

\section{Convex reformulation of the outer problem} \label{sec:outer-convex}
As mentioned in the previous section, the outer problem has a bilinear term in the objective function. Also, the value of $\bm{Pr}_s$ that defines the probability of an $N$-$k$ scenario $s \in \mathcal S$ depends on the $\bm k$ transmission lines that are being interdicted in scenario $s$. Generally, it is computationally inefficient to maintain a look-up table that calculates the value of $\bm{Pr}_s$ for each scenario. Instead, we model the relationship between $\bm{Pr}_s$ and lines interdicted in the $N$-$k$ scenario $s$ explicitly by using binary interdiction decision variables $x_{ij}$ for each line $(i,j)$. We then use logarithms to address both these issues and convert the outer problem to a maximization of a concave function. After taking logarithms, the outer problem is as follows:
\begin{flalign}
& \max_{s \in \mathcal S} \quad \log \bm{Pr}_s + \log z_s \quad \text{ subject to:} & \notag \\
& \log \bm{Pr}_s = \sum_{(i,j) \in \mathcal E} \left(\log \bm{Pr}_{ij} \cdot x_{ij}\right), &\label{eq:nk-ps} \\
& \sum_{(i,j) \in \mathcal E} x_{ij} = \bm k, & \label{eq:nk-minlp} \\
& x_{ij} \in \{0,1\} \quad \forall (i,j) \in \mathcal E, & \label{eq:x-minlp}  \\
& z_s \geqslant 0 \quad \forall s \in \mathcal S. & \label{eq:zs} 
\end{flalign}
The decision variables in this problem are $x_{ij}$, $z_s$, and $\bm{Pr}_s$. Once the $N$-$k$ scenario $s\in \mathcal S$ is fixed, the values of all the decision variables are completely determined. In particular, the value of $z_s$ is determined by the the inner problem, given by \eqref{eq:obj-inner-minlp} -- \eqref{eq:acopf2}. Hence, the maximization is over the set of all scenarios $s \in \mathcal S$. Constraints \eqref{eq:nk-ps}, \eqref{eq:nk-minlp}, and \eqref{eq:x-minlp} define the value of $\bm{Pr}_s$ based on the values taken by the interdiction variables $x_{ij}$ corresponding to the $N$-$k$ scenario $s \in \mathcal S$. The assumption that the line-failure probabilities $\bm{Pr}_{ij}$ are independent is used to express $\bm{Pr}_s$ using $\bm{Pr}_{ij}$s in constraint \eqref{eq:nk-ps}. We now let $p =\log \bm{Pr}_s$ and the  optimization problem is equivalently written as:
\begin{flalign}
& \max \quad  p + \log z_s \quad \text{ subject to:} &  \label{eq:outer-concave}\\
& p = \sum_{(i,j) \in \mathcal E} \left(\log \bm{Pr}_{ij} \cdot x_{ij}\right), \text{ \eqref{eq:nk-minlp}, \eqref{eq:x-minlp}, and \eqref{eq:zs}. } & \notag
\end{flalign}
The function $p + \log z_s$ in the objective is jointly concave when $p, z_s > 0$. Hence, the outer problem is a convex optimization problem, provided $p, z_s > 0$.

In the next section, we present convex relaxations and linear approximations for the inner problem in the PNK. The inner problem, as presented in Section \ref{sec:minlp}, is nonlinear because of the AC power flow equations in \eqref{eq:acopf2}. This makes computing a global optimum, even for reasonably sized test instances, computationally intractable \cite{Chen2016}. As a result, there has been recent work developing convex relaxations in the literature. The relaxations include the semidefinite programming \cite{Bai2008}, the second-order cone (SOC) \cite{Jabr2006}, the quadratic convex \cite{Hijazi2017}, and the moment-based relaxations \cite{Molzahn2014}. Each of these relaxations has an associated solution quality/computation time trade-off.  In this paper we use the SOC relaxation because it has a favorable balance of computation time and solution quality in this application domain. 
We also present two approximations to the AC power flow equations: the linearized DC power flow formulation and a traditional network flow formulation. 
It is important to note that the algorithm presented in this paper is generic and can be applied to any convex relaxations of the AC power flow equations. 

\section{Convex relaxation and linear approximations of the inner problem} \label{sec:inner}
\subsection{The SOC relaxation} \label{subsec:soc}
First we present the SOC relaxation of the AC power flow equations in \eqref{eq:acopf2} to obtain a SOC program for the inner problem for a fixed $N$-$k$ scenario. We begin by rewriting the constraints for the inner problem in the following equivalent lifted form:
\begin{flalign}
& W^s_{ij} = V_i^s V_j^{s*} \quad \forall (i, j) \in \mathcal E_s, & \label{eq:acopf-w1} \\
& (\bm v_i^l)^2 \leqslant W_{ii}^s \leqslant (\bm v_i^u)^2 \quad \forall i\in \mathcal N, & \label{eq:vmag-minlp-w} \\
&\tan(-\bm \theta_{ij}^{\Delta})  \leqslant \frac{\operatorname{Im}(W_{ij}^s)}{\operatorname{Re}(W_{ij}^s)} \leqslant \tan(\bm \theta_{ij}^{\Delta}) \quad \forall (i,j) \in \mathcal E_s, & \label{eq:theta-minlp-w} \\
& S_{ij}^s = \bm Y_{ij}^* W_{ii}^{s} - \bm Y_{ij}^* W_{ij}^s \quad \forall (i,j) \in \mathcal E_s \bigcup \mathcal E_s^r, & \label{eq:acopf2-w} \\
& \text{ and constraints \eqref{eq:pg-minlp}, \eqref{eq:qg-minlp}, \eqref{eq:thermal-minlp}, \eqref{eq:loadshed-minlp} and \eqref{eq:acopf1}.} \notag  
\end{flalign} 

The introduction of these auxiliary variables ensures that the nonlinearity in the inner problem is isolated in constraint \eqref{eq:acopf-w1}. Given the above set of lifted constraints for the inner problem, the SOC relaxation is obtained by replacing constraint \eqref{eq:acopf-w1} with $\lvert W_{ij}^s \rvert^2 \leqslant W_{ii}^s W_{jj}^s$ for every $(i, j) \in \mathcal E_s$. For the sake of completeness, the full formulation for the PNK, with the convex reformulation of the outer problem and the SOC relaxation of the inner problem, is shown below:
\begin{flalign*}
& (\mathcal F_{\operatorname{soc}}):  \quad \max \quad  p + \log z_s & \\
& \text{ subject to:} \,\, p = \sum_{(i,j) \in \mathcal E} \left(\log \bm{Pr}_{ij} \cdot x_{ij}\right), \text{ \eqref{eq:nk-minlp}, and \eqref{eq:x-minlp}} & \notag
\end{flalign*}
where $z_s$ for every $N$-$k$ scenario $s \in \mathcal S$ is defined by
\begin{flalign}
& z_s = \min \quad \sum_{i \in {\mathcal N}} \operatorname{Re}(\bm S^d_i) \ell_i^s & \notag \\
& \text{ subject to:} & \notag \\
& \lvert W_{ij}^s \rvert^2 \leqslant W_{ii}^s W_{jj}^s \quad \forall (i, j) \in \mathcal E_s, & \label{eq:acopf-w1a} \\
& \text{\eqref{eq:vmag-minlp-w}, \eqref{eq:theta-minlp-w}, \eqref{eq:acopf2-w}, \eqref{eq:pg-minlp}, \eqref{eq:qg-minlp}, \eqref{eq:thermal-minlp}, \eqref{eq:loadshed-minlp}, and \eqref{eq:acopf1}.} \notag  
\end{flalign} 

\subsection{The DC approximation} \label{subsec:dc}
In this section, we present the linearized DC approximation of the AC power flow equations in \eqref{eq:acopf2}; we note that the DC approximation results in the inner problem being a linear program for a fixed $N$-$k$ scenario $s \in \mathcal S$. To make the presentation of the DC approximation simple, it is useful to represent the AC power flow equations in their polar form. Eq. \eqref{eq:acopf2}, in polar coordinates, is given by the following set of two equations for every $(i,j) \in \bar{\mathcal E}_s$:
\begin{flalign*}
&p_{ij}^s = (v_i^s)^2\bm g_{ij} - v_i^s v_j^s \left(\bm g_{ij} \cos(\theta_i^s - \theta_j^s) + \bm b_{ij} \sin(\theta_i^s - \theta_j^s)\right), & \\
&q_{ij}^s = -(v_i^s)^2\bm b_{ij} - v_i^s v_j^s \left(\bm g_{ij} \sin(\theta_i^s - \theta_j^s) - \bm b_{ij} \cos(\theta_i^s - \theta_j^s)\right). & 
\end{flalign*}
Three basic assumptions are used to derive the DC approximation of the AC power flow constraints in Cartesian form. They are as follows:
\begin{enumerate}[label=(A\arabic*)]
    \item $|\bm b_{ij}| \gg |\bm g_{ij}|$ for every line $(i,j)$.
    \item The voltage magnitude at each node or bus is $1$ pu, \textit{i.e.}, $v_i^s = 1$ pu for each bus $i$.
    \item The voltage angle difference $(\theta^s_i - \theta^s_j)$ across any line $(i, j) \in \mathcal E$ is small enough that $\cos (\theta_i^s - \theta_j^s) \approx 1 $ and $\sin (\theta_i^s - \theta_j^s) \approx  (\theta_i^s - \theta_j^s).$
\end{enumerate}
Using these three assumptions on the AC power flow equations yields
\begin{flalign}
& p_{ij}^s = -\bm b_{ij} (\theta^s_i - \theta^s_j) \text{ and } q_{ij}^s = 0 \quad \forall (i,j) \in \mathcal E_s \bigcup \mathcal E_s^r. & \label{eq:dcopf}
\end{flalign}
Applying Eq. \eqref{eq:dcopf} in the inner problem for every scenario $s\in \mathcal S$ and the convex reformulation of the outer problem yields the following DC approximation of the PNK:
\begin{flalign*}
& (\mathcal F_{\operatorname{dc}}):  \quad \max \quad  p + \log z_s & \\
& \text{ subject to:} \,\, p = \sum_{(i,j) \in \mathcal E} \left(\log \bm{Pr}_{ij} \cdot x_{ij}\right), \text{ \eqref{eq:nk-minlp}, and \eqref{eq:x-minlp}} & \notag
\end{flalign*}
where
\begin{flalign}
& z_s = \min \quad \sum_{i \in {\mathcal N}} \bm p^d_i \ell_i^s \quad \text{ subject to:}& \label{eq:obj-inner-dc} \\
& \bm p^{gl}_i \leqslant p^{gs}_i \leqslant \bm p^{gu}_i \quad \forall i\in \mathcal N, & \label{eq:pg-dc} \\
& \lvert p_{ij}^s \rvert \leqslant \bm t_{ij} \quad \forall (i,j) \in \mathcal E_s, & \label{eq:thermal-dc} \\
%& -\bm \theta_{ij}^{\Delta} \leqslant (\theta_i^s - \theta_j^s) \leqslant \bm \theta_{ij}^{\Delta} \quad \forall (i,j) \in \mathcal E_s, & \label{eq:theta-dc} \\
& 0 \leqslant \ell_i^s \leqslant 1 \quad \forall i \in \mathcal N, & \label{eq:loadshed-dc} \\
& p_i^{gs} - (1-\ell_i^s)\bm p_i^d = \sum_{(i,j) \in \mathcal E_s \bigcup \mathcal E_s^r} p_{ij}^s \quad \forall i\in \mathcal N, \text{ and} & \label{eq:dc1} \\
& p_{ij}^s = -\bm b_{ij} (\theta^s_i - \theta^s_j) \quad \forall (i,j) \in \mathcal E_s \bigcup \mathcal E_s^r. & \label{eq:dc2}
\end{flalign}
In the next section, we present another linear approximation of the inner problem of the PNK: a network flow (NF) approximation. We note that the NF approximation is actually a relaxation of the DC approximation presented in this section. 

\subsection{NF approximation} \label{subsec:nf}
In this section, we present a relaxation of the DC approximation using capacitated flows in networks. The NF approximation for the PNK is obtained by dropping constraint \eqref{eq:dc2} from the DC approximation presented in the previous section. It is trivial to observe that for a given $N$-$k$ scenario, the inner problem without \eqref{eq:dc2} is a maximum flow problem. For the sake of completeness, the PNK with the NF approximation is shown below:
\begin{flalign*}
& (\mathcal F_{\operatorname{nf}}):  \quad \max \quad  p + \log z_s & \\
& \text{ subject to:} \,\, p = \sum_{(i,j) \in \mathcal E} \left(\log \bm{Pr}_{ij} \cdot x_{ij}\right), \text{ \eqref{eq:nk-minlp}, and \eqref{eq:x-minlp}} & \notag
\end{flalign*}
where
\begin{flalign*}
& z_s = \min \quad \sum_{i \in {\mathcal N}} \bm p^d_i \ell_i^s \text{ subject to: \eqref{eq:pg-dc}, \eqref{eq:thermal-dc}, \eqref{eq:loadshed-dc},  and \eqref{eq:dc1}}. &
\end{flalign*}

\section{A cutting-plane algorithm} \label{sec:cp}
In this section, we present a generic cutting-plane algorithm that is common to the MINLP, the SOC relaxation, and the DC and NF approximation for the PNK. The cutting-plane algorithm uses the Stackelberg game structure that is inherent to all the formulations. The main difficulty in developing any algorithm for the PNK arises from the nonconvex max-min nature of the problem. A number of techniques have been studied and analyzed in the literature to convert bilevel max-min problems like the PNK to a single mixed-integer program (see  \cite{Alvarez2004, Motto2005}). These techniques are known to have scaling issues for large-scale test instances. To understand why scaling is a problem, it is useful to think about an $N$-$k$ contingency scenario $s$ as modifying the original set of lines $\mathcal E$ to be $\mathcal E_s$. Representing this behavior as constraints is key to converting the bilevel program to a single-level mixed-integer program and can be achieved by using disjunctions. Furthermore, if the $N$-$k$ scenario includes node or substation failures, this further complicates the model with additional multilinear terms \cite{Salmeron2009}. Instead, we present a cutting-plane algorithm that, instead of converting the bilevel PNK to a single-level problem, generates cutting planes (using the solution of the inner problem) that are added sequentially to the outer problem. 

The algorithm relies on constructing a sequence of piecewise linear functions that upper bound the total active load shed $z_s$, given by the solution to the inner problem or its relaxation and approximations, for any $s \in \mathcal S$. For any $N$-$k$ scenario $\hat s \in \mathcal S$ and the associated solution $x_{ij}$ for $(i,j) \in \mathcal E$, let $z_{\hat s}$ denote the minimum load shed that is obtained by solving the inner problem given by \eqref{eq:obj-inner-minlp} -- \eqref{eq:acopf2} or one of its convex relaxation or approximations in Section \ref{sec:inner}. Then, the algorithm computes coefficients $\alpha_{ij}(\hat s)$ for each line $(i,j) \in \mathcal E$ such that 
\begin{flalign}
&z_s \leqslant z_{\hat s} + \sum_{(i,j) \in \mathcal E} \alpha_{ij}(\hat s) \cdot x_{ij} \quad \forall s \in \mathcal S. & \label{eq:woodscut}
\end{flalign}
The cut in \eqref{eq:woodscut} has been developed and found to be effective for the deterministic variant of formulation $\mathcal F_{\operatorname{dc}}$ (see \cite{Salmeron2009}). In this section, we show that this cut generalizes to the PNK and to different relaxations and approximations of the AC power flow equations. The inequality in \eqref{eq:woodscut} provides an upper bound for $z_s$, the objective function of the inner problem, for every feasible $N$-$k$ scenario $s \in \mathcal S$. The linear cut in \eqref{eq:woodscut} is very general, and there are many choices for the cut coefficients $\alpha_{ij}(\hat s)$. The key challenge is to choose \textit{tight} values for each coefficient that do not remove any feasible $N$-$k$ scenario. For the PNK, these coefficients $\alpha_{ij}(\hat s)$ are computed using a combination of the subproblem solution for scenario $s$ and the physics of power flow in the network. Using the inequality in \eqref{eq:woodscut}, the PNK is equivalently written as:
\begin{flalign*}
& (\mathcal F_{\operatorname{cp}}) \quad \max \quad p + \log z_s \quad \text{ subject to:} \quad & \\
&z_s \leqslant z_{\hat s} + \sum_{(i,j) \in \mathcal E} \alpha_{ij}(\hat s) \cdot x_{ij} \quad \forall \hat s \in \mathcal S, & \\
& p = \sum_{(i,j) \in \mathcal E} \left(\log \bm{Pr}_{ij} \cdot x_{ij}\right), \text{ \eqref{eq:nk-minlp}, and \eqref{eq:x-minlp}.} & \notag
\end{flalign*}
Pseudocode of the cutting-plane algorithm for the PNK using formulation $\mathcal F_{\operatorname{cp}}$ is given in Algorithm \ref{algo:pseudocode}. The pseudocode assumes a black box for solving the inner problem and computing the cut coefficients in \eqref{eq:woodscut}.

\begin{algorithm}
\caption{Cutting-plane algorithm: pseudocode}\label{algo:pseudocode}
\begin{algorithmic}[1]
\vspace{1ex}
\Input optimality tolerance, $\varepsilon > 0$
\Output $s^* \in \mathcal S$, an $\varepsilon$-optimal $N$-$k$ scenario to the PNK
\State initial problem: $\mathcal F_{\operatorname{cp}}$ without constraint \eqref{eq:woodscut}
\State $f^* \gets -\infty$ \Comment{lower bound on the optimal obj. value}
\State $\bar f\gets +\infty$ \Comment{upper bound on the optimal obj. value} 
\State $\hat s \gets $ any initial $N$-$k$ scenario
\State $p(\hat s) \gets \sum_{(i,j) \in \tilde{\mathcal E}_{\hat s}} \log \bm{Pr}_{ij}$
\State \label{step:sub}solve subproblem for $\hat s$ and let $z_{\hat s}$ be the objective value
\State $f(\hat s) \gets p(\hat s) + \log z_{\hat s}$
\If{$f(\hat s) > f^*$} $f^* \gets f(\hat s)$ and $s^* \gets \hat s$ 
\EndIf
%\If{$\bar f - f^* \leqslant \varepsilon f^*$} $(s^*, f^*)$ is the $\varepsilon$-optimal solution to the PNK, stop 
%\EndIf
\State compute $\alpha_{ij}(\hat s)$ for every $(i,j) \in \mathcal E$ satisfying \eqref{eq:woodscut} 
\State \label{step:master}add $z_s \leqslant z_{\hat s} + \sum_{(i,j) \in \mathcal E} \alpha_{ij}(\hat s) \cdot x_{ij}$ to $\mathcal F_{\operatorname{cp}}$ and resolve
\State update $\hat s$, and set $\bar f$ using solution from Step \ref{step:master}
\If{$\bar f - f^* \leqslant \varepsilon f^*$} $(s^*, f^*)$ is the $\varepsilon$-optimal solution to the PNK, stop  
\EndIf
\State return to step: \ref{step:sub}
\end{algorithmic}
\end{algorithm}

\subsection{Computing cut coefficients}
Here we present an algorithm to compute the cut coefficients $\alpha_{ij}(\hat s)$ for a given $N$-$k$ scenario $\hat s$ and for every $(i,j) \in \mathcal E$ using the solution to the inner problem. First we present an expression for the cut coefficients, and then we present a proof of validity of the cut when the NF approximation is used for the inner problem followed by the necessary conditions under which the cut is still valid for the DC approximation and the SOC relaxation of the inner problem.  Regardless of the formulation that is used for the inner problem ($\mathcal F_{\operatorname{soc}}, \mathcal F_{\operatorname{dc}}$, or $\mathcal F_{\operatorname{nf}}$), the solution to any of these problems for a fixed $N$-$k$ scenario $\hat s$ provides a value of $p_{ij}^{\hat s}$ for each line $(i,j) \in \mathcal E_{\hat s} \bigcup \mathcal E^r_{\hat s}$. Each $p_{ij}^{\hat s}$ represents the active power flowing on line $(i,j) \in \mathcal E_{\hat s} \bigcup \mathcal E^r_{\hat s}$ during the $N$-$k$ scenario $\hat s$. Given these values of $p_{ij}^{\hat s}$, the cut coefficients are defined by the following equation:
\begin{flalign}
\alpha_{ij}^{\hat s} = \begin{cases} \max\{|p_{ij}^{\hat s}|, |p_{ji}^{\hat s}|\} & \mbox{if } (i,j) \in \mathcal E \setminus \tilde{\mathcal E}_{\hat s}, \\ 
0 & \text{ otherwise.} \end{cases} \label{eq:coefficients}
\end{flalign}
Recall that $\tilde{\mathcal E}_{\hat s}$ is the set of lines that is removed for the $N$-$k$ scenario $\hat s$. When the coefficient $\alpha_{ij}^{\hat s}$ provides an upper bound on the total amount of active load that must be shed as the line $(i, j)$ is removed from the network during an $N$-$k$ scenario $\hat s$, the cut given by \eqref{eq:woodscut} is valid. The following theorem proves that the cut in \eqref{eq:woodscut} with the cut coefficients in \eqref{eq:coefficients} is valid for formulation $\mathcal F_{\operatorname{nf}}$, which uses the NF approximation of the inner problem. 

\begin{theorem} \label{thm:nf}
    For any $N$-$k$ scenario $\hat s \in \mathcal S$, the cut $z_s \leqslant z_{\hat s} + \sum_{(i,j) \in \mathcal E} \alpha_{ij}(\hat s) \cdot x_{ij}$, with $\alpha_{ij}(\hat s)$ taking the values in Eq. \eqref{eq:coefficients}, is valid for formulation $\mathcal F_{\operatorname{nf}}$.
\end{theorem}
\begin{proof}
As discussed in Section \ref{subsec:nf}, the inner problem in formulation $\mathcal F_{\operatorname{nf}}$ is a maximum flow problem because minimizing the active load shed in the system is equivalent to maximizing the active load satisfied in the system. Maximizing the amount of active load satisfied in the system subject to the NF constraints in $\mathcal F_{\operatorname{nf}}$ is simply the maximum flow problem in a directed graph; $p_{ij}^{\hat s}$ represents the amount of active power flowing through line $(i, j) \in \mathcal E_{\hat s} \bigcup \mathcal E_{\hat s}^r$. At the optimal solution of the inner problem of $\mathcal F_{\operatorname{nf}}$ for a fixed $s \in \mathcal S$, the maximum amount of load that can be shed when line $(i,j)$ is removed from the network is given exactly by the expression in \eqref{eq:coefficients}. This follows from the max-flow min-cut theorem \cite{Cormen2009}.
\end{proof}

We note that a trivial value of $\alpha_{ij}(\hat s) = \sum_{i \in \mathcal N} \bm p_i^d$ for every $(i,j) \in \mathcal E \setminus \tilde{\mathcal E}_{\hat s}$ (total active power demand in the system) results in a valid cut for the linearized DC approximation and the SOC relaxation in Sections \ref{subsec:dc} and \ref{subsec:soc}, respectively. But, this value for cut coefficients is very weak and leads to poor convergence of the cutting-plane algorithm. This motivates the use of heuristic choices for  cut-coefficients that have better convergence properties. As a result, the value of the cut coefficients, as given by Eq. \eqref{eq:coefficients}, is valid for the formulation $\mathcal F_{\operatorname{nf}}$. For the DC approximation of the PNK presented in Section \ref{subsec:dc}, the cut given in \eqref{eq:woodscut} using the coefficients in \eqref{eq:coefficients} is not necessarily valid and can remove feasible solutions to the PNK. An interested reader is referred to \cite{Baldick2006,Salmeron2009} for a counterexample. Nevertheless, the following theorem mathematically formalizes an empirical statement in \cite{Salmeron2009} concerning the validity of the cut in \eqref{eq:woodscut} using the coefficients in \eqref{eq:coefficients}. This theorem specifies a necessary condition under which the cut given by \eqref{eq:woodscut} is valid for the formulations $\mathcal F_{\operatorname{dc}}$ and $\mathcal F_{\operatorname{soc}}$. 

\begin{theorem} \label{thm:dc_soc}
    For any $N$-$k$ scenario $\hat s \in \mathcal S$, the cut $z_s \leqslant z_{\hat s} + \sum_{(i,j) \in \mathcal E} \alpha_{ij}(\hat s) \cdot x_{ij}$, with $\alpha_{ij}(\hat s)$ taking the values in Eq. \eqref{eq:coefficients}, is valid for the formulations $\mathcal F_{\operatorname{dc}}$ and $\mathcal F_{\operatorname{soc}}$ if the following condition is satisfied:
    \begin{enumerate}[label=(B\arabic*)]
        \item The removal of the set of lines in $(i,j) \in \mathcal E_{\hat s}$, each carrying $p_{ij}^{\hat s}$ and $p_{ji}^{\hat s}$ units of real power during contingency $\hat s$ from $i$ to $j$ and $j$ to $i$, respectively, leads to an additional total real power load shedding of at most $\sum_{(i,j) \in \mathcal E_{\hat s}}  \max\{|p_{ij}^{\hat s}|, |p_{ji}^{\hat s}|\}$.
    \end{enumerate}
\end{theorem}
\begin{proof}
The cut in \eqref{eq:woodscut} can be interpreted as follows: given a $N$-$k$ scenario $\hat s \in \mathcal S$ with an inner problem objective value $z_{\hat s}$, removal of a set of lines $(i,j) \in \mathcal E_{\hat s}$ would result in an increase in load shed by a value of at most $\sum_{(i,j) \in \mathcal E_{\hat s}} \alpha_{ij}$. Eq. \eqref{eq:coefficients} and the assumption (B1) together indicate that $$\sum_{(i,j) \in \mathcal E_{\hat s}} \alpha_{ij} = \sum_{(i,j) \in \mathcal E_{\hat s}}  \max\{|p_{ij}^{\hat s}|, |p_{ji}^{\hat s}|\},$$ completing the proof.  
\end{proof}

% Nevertheless, under the following necessary condition, the cut is valid.
% \begin{enumerate}[label=(B\arabic*)]
% %\item the contingency $\hat s$ does not lead to cascading failures and 
% \item The removal of the set of lines in $(i,j) \in \mathcal E_{\hat s}$, each carrying $p_{ij}^{\hat s}$ units of real power during contingency $\hat s$, leads to a total real power load shedding of at most $\sum_{(i,j) \in \mathcal E_{\hat s}}  p_{ij}^{\hat s}$.
% \end{enumerate}
In general, assumption (B1) need not hold (see \cite{Baldick2006}) because of Braess's paradox as exhibited in electric transmission systems. The validity of assumption (B1) has been empirically observed for many available power system test cases via extensive computational experiments (see \cite{Salmeron2009}) for the deterministic variant of the DC approximation. As for the SOC relaxation of the PNK given by formulation $\mathcal F_{\operatorname{soc}}$, we make assumption (B1); in fact, the sufficiency condition for the validity of the cut given in \eqref{eq:woodscut} and \eqref{eq:coefficients} for models $\mathcal F_{\operatorname{soc}}$ and $\mathcal F_{\operatorname{dc}}$ is still an open question. This makes the cutting-plane algorithm a heuristic for formulations $\mathcal F_{\operatorname{soc}}$ and $\mathcal F_{\operatorname{dc}}$. 

\subsection{Logical constraints} \label{subsec:logic}
We now present an additional logical constraint that is added to speed up the convergence of the cutting-plane algorithm. 
Given a feasible $N$-$k$ scenario $\hat s\in \mathcal S$, the constraint is given by 
\begin{flalign}
&  \sum_{(i,j) \in \tilde{\mathcal E}_{\hat s}} x_{ij} \leqslant \bm k - 1 \label{eq:logic}. 
% \sum_{(i,j) \in \mathcal E_{\hat s}} (1-x_{ij}) 
\end{flalign}
This constraint forces any new $N$-$k$ scenario to differ from $\hat s$ in at least one component. This type of cut has been observed to be effective in similar interdiction problems \cite{Israeli2002,Salmeron2009}. In the following theorem, we prove that the cutting-plane algorithm, along with the logical constraints in \eqref{eq:logic}, converges in a finite number of iterations.

\begin{theorem} \label{thm:convergence}
The cutting-plane algorithm presented in Algorithm \ref{algo:pseudocode}, along with the logical constraints in \eqref{eq:logic}, converges in finite number of iterations.
\end{theorem}
\begin{proof}
 We start by observing that adding just the logical constraints in \eqref{eq:logic} (and neglecting the cuts in \eqref{eq:woodscut}) to $\mathcal F_{\operatorname{cp}}$ makes the cutting-plane algorithm a complete enumeration algorithm. This itself ensures finite convergence because there are a finite number of $N$-$k$ scenarios. \eqref{eq:woodscut}, if valid (for $\mathcal F_{\operatorname{nf}}$ or for $\mathcal F_{\operatorname{dc}}$ and $\mathcal F_{\operatorname{soc}}$ under the necessary conditions given by Theorem \ref{thm:dc_soc}), ensures that not all feasible solutions are enumerated and accelerates the convergence. If the cut is not valid, then using results in a suboptimal solution does not affect the finite convergence property, completing the proof.
 \end{proof}

\section{Computational results} \label{sec:results}
In this section, we demonstrate the computational effectiveness of the cutting-plane algorithm and heuristics for computing an optimal/feasible solution to the linear approximations and relaxation of the PNK. We use five test instances for all our simulations, classified into the following four categories: 
\begin{enumerate}[label=(C\arabic*)]
    \item Small test instances: IEEE 14-bus and IEEE single-area RTS96 (with 24 buses) test systems. 
    \item Medium test instances: IEEE three-area RTS96 (with 73 buses) and IEEE 118-bus test systems.
    \item Large test instances: WECC 240 test system. We note that the WECC 240 is an aggregated model of a real transmission system in the western United States.
    \item Very large test instances: PEGASE 1354-bus and the Polish 2383-bus Winter peak test systems.
\end{enumerate}
All computational experiments were run on an Intel Haswell 2.6 GHz, 62 GB, 20-core supercomputing machine at Los Alamos National Laboratory. The algorithms were implemented using PowerModels \cite{Coffrin2017}, an open-source framework for implementing power flow formulations. In the interest of using state-of-the-art commercial solvers, which support only convex quadratically constrained quadratic programs, a linear outer approximation of the outer problem objective function is used in place of its concave counterpart. CPLEX 12.7 is used for solving all the mixed-integer linear and convex optimization problems presented in this paper. The data on the probability of line failures are available as a part of reliability data for the IEEE single-area and three-area test cases \cite{Grigg1999}.  For the remaining three test cases, the probability of line failure for each line is generated using the following procedure. We compute the maximum and minimum probability ($p_{\max}$ and $p_{\min}$, respectively) among all the failure probabilities for the IEEE single-area RTS96 test case and generate the probability of line failures uniformly in the interval $[p_{\min}, p_{\max}]$ for the other test cases. This data and the code are made available in \cite{github}. 

In order to ensure that the solutions obtained using the cutting plane algorithm on $\mathcal F_{\operatorname{soc}}$, $\mathcal F_{\operatorname{dc}}$, and $\mathcal F_{\operatorname{nf}}$ are meaningful for the PNK, the nonlinear inner problem in $\mathcal F_{\operatorname{ac}}$ is solved using the $N$-$k$ scenario obtained from the cutting-plane algorithm using Ipopt \cite{ipopt} as the local solver. This process of obtaining a feasible solution to the PNK using its MINLP formulation, $\mathcal F_{\operatorname{ac}}$, is also called primal recovery or feasible solution recovery.

\subsection{Comparison of solution with total enumeration} \label{sec:enum}
In this section, we concentrate on only the small instances in category (C1) where an optimal solution to formulations $\mathcal F_{\operatorname{soc}}$, $\mathcal F_{\operatorname{dc}}$, and $\mathcal F_{\operatorname{nf}}$ can be obtained by complete enumeration of all the $N$-$k$ scenarios. 
% As observed in Section \ref{sec:cp}, the cutting-plane algorithm is exact for the NF-based formulation, $\mathcal F_{\operatorname{nf}}$, of the PNK (see Theorem \ref{thm:nf}); this is due to the objective values obtained by the cutting-plane algorithm on $\mathcal F_{\operatorname{nf}}$ matching with those obtained by total enumeration. Despite the cutting-plane algorithm being a heuristic for the SOC relaxation and the DC approximation of the PNK, we observe that the objective values obtained by the cutting-plane algorithm and using total enumeration match for the SOC relaxation and the DC approximation of the PNK.
Here, we compare the solutions obtained by the three formulations with their corresponding optimal solutions obtained via complete enumeration of all possible feasible solutions. We restrict our attention to values of $\bm k$ in $\{2, 3, 4\}$. 

\begin{table}[!h]
    \centering
    \begin{tabular}{cccc}
        \toprule
        \# buses & $\bm k$ & cutting-plane & complete enumeration \\
                 &         & obj. value (MW) & obj. value (MW)    \\
        \midrule
        14 & 2 & 23.99 & 23.99\\
        14 & 3 & 11.51 & 11.51\\
        14 & 4 & 7.47 & 7.47\\
        24 & 2 & 28.75 & 28.75\\
        24 & 3 & 15.52 & 15.52\\
        24 & 4 & 20.48 & 20.48\\
        \bottomrule 
    \end{tabular}
    \caption{Comparison of the objective values obtained by using the cutting-plane algorithm on $\mathcal F_{\operatorname{nf}}$ (NF approximation) and the optimal solution obtained via complete enumeration.}
    \label{tab:nf_enum}
\end{table}

\begin{table}[!h]
    \centering
    \begin{tabular}{cccc}
        \toprule
        \# buses & $\bm k$ & cutting-plane & complete enumeration \\
                 &         & obj. value (MW) & obj. value (MW)    \\
        \midrule
        14 & 2 & 23.99 & 23.99\\
        14 & 3 & 11.51 & 11.51\\
        14 & 4 & 7.47 & 7.47\\
        24 & 2 & 28.75 & 28.75\\
        24 & 3 & 15.52 & 15.52\\
        24 & 4 & 20.48 & 20.48\\
        \bottomrule 
    \end{tabular}
    \caption{Comparison of the objective values obtained by using the cutting-plane algorithm on $\mathcal F_{\operatorname{dc}}$ (DC approximation) and the optimal solution obtained via complete enumeration.}
    \label{tab:dc_enum}
\end{table}

\begin{table}[!h]
    \centering
    \begin{tabular}{cccc}
        \toprule
        \# buses & $\bm k$ & cutting-plane & complete enumeration \\
                 &         & obj. value (MW) & obj. value (MW)    \\
        \midrule
        14 & 2 & 24.02 & 24.02\\
        14 & 3 & 11.53 & 11.53\\
        14 & 4 & 7.47 & 7.47\\
        24 & 2 & 28.75 & 28.75\\
        24 & 3 & 19.18 & 19.18\\
        24 & 4 & 20.97 & 20.97\\
        \bottomrule 
    \end{tabular}
    \caption{Comparison of the objective values obtained by using the cutting-plane algorithm on $\mathcal F_{\operatorname{soc}}$ (SOC relaxation) and the optimal solution obtained via complete enumeration.}
    \label{tab:soc_enum}
\end{table}

It can be observed from Tables \ref{tab:dc_enum} and \ref{tab:soc_enum} that the cutting-plane algorithm is able to compute an optimal solution to formulations $\mathcal F_{\operatorname{dc}}$ and $\mathcal F_{\operatorname{soc}}$, despite being a heuristic, for the instances in category (C1) for values of $\bm k \in \{2,3,4\}$. The results in Table \ref{tab:nf_enum} are not surprising because the cutting-plane algorithm is indeed an exact algorithm for the PNK with the NF approximation. Figs. \ref{fig:iterations_cp_14} and \ref{fig:iterations_cp_24} present the number of iterations of the cutting-plane algorithm for the SOC relaxation, the DC approximation, and the NF approximation in the IEEE 14-bus system and the IEEE single-area RTS96 system with 24 buses, respectively. A consistent trend that emerges from Figs. \ref{fig:iterations_cp_14} and \ref{fig:iterations_cp_24} is that the number of iterations of the cutting-plane algorithm is highest for the NF approximation and lowest for the SOC relaxation of the PNK, except for in the IEEE 14-bus system when $\bm k = 2$. The computation time is not reported for these small instances because the cutting-plane algorithm converged to an optimal solution within $5$ seconds for all the runs.

\begin{figure}[!h]
	\centering
	\begin{tikzpicture}[scale=0.8]
\begin{axis}[
    ybar,
    enlarge x limits=0.25,
    legend style={at={(0.7,0.85)},draw=none,fill=none,
      anchor=north,legend columns=-1},
    ylabel={\# of iterations},
    symbolic x coords={$\bm k=2$, $\bm k=3$,$\bm k=4$,$\bm k=5$},%,$\bm k=6$},
    xtick=data,
    ymajorgrids,
    %nodes near coords,
    %nodes near coords align={vertical},
    ]
\addplot
	coordinates {($\bm k=2$, 7) ($\bm k=3$, 10) ($\bm k=4$, 4) ($\bm k=5$, 5)}; %($\bm k=6$, 4)};
\addplot[fill=red!30, postaction={pattern=horizontal lines}]
    coordinates {($\bm k=2$, 4) ($\bm k=3$, 6) ($\bm k=4$, 3) ($\bm k=5$, 3)}; %($\bm k=6$, 3)};
\addplot[fill=brown!30, postaction={pattern=north east lines}]
    coordinates {($\bm k=2$, 5) ($\bm k=3$, 5) ($\bm k=4$, 2) ($\bm k=5$, 2)}; %($\bm k=6$, 2)};
\legend{$\mathcal F_{\operatorname{nf}}$, $\mathcal F_{\operatorname{dc}}$, $\mathcal F_{\operatorname{soc}}$}
\end{axis}
\end{tikzpicture}
\caption{Number of iterations of the cutting-plane algorithm applied to formulations $\mathcal F_{\operatorname{soc}}$, $\mathcal F_{\operatorname{dc}}$, and $\mathcal F_{\operatorname{nf}}$ in the IEEE 14-bus system.}
\label{fig:iterations_cp_14}
\end{figure}
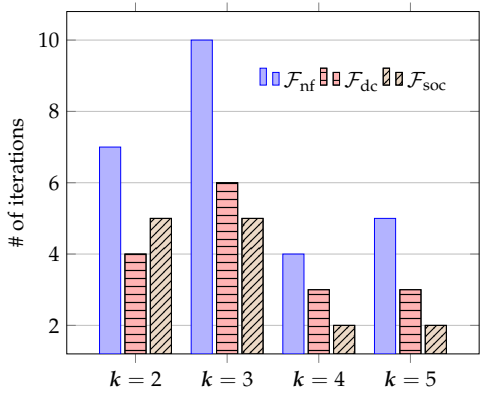

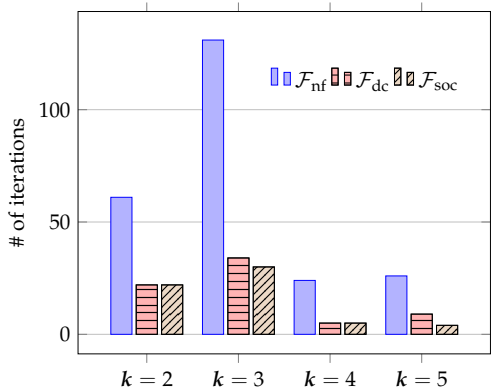
\begin{figure}
	\centering
	\begin{tikzpicture}[scale=0.8]
\begin{axis}[
    ybar,
    enlarge x limits=0.25,
    legend style={at={(0.7,0.85)},draw=none,fill=none,
      anchor=north,legend columns=-1},
    ylabel={\# of iterations},
    symbolic x coords={$\bm k=2$, $\bm k=3$,$\bm k=4$,$\bm k=5$},%,$\bm k=6$},
    xtick=data,
    ymajorgrids,
    %nodes near coords,
    %nodes near coords align={vertical},
    ]
\addplot
	coordinates {($\bm k=2$, 61) ($\bm k=3$, 131) ($\bm k=4$, 24) ($\bm k=5$, 26)}; %($\bm k=6$, 21)};
\addplot[fill=red!30, postaction={pattern=horizontal lines}]
    coordinates {($\bm k=2$, 22) ($\bm k=3$, 34) ($\bm k=4$, 5) ($\bm k=5$, 9)}; %($\bm k=6$, 10)};
\addplot[fill=brown!30, postaction={pattern=north east lines}]
    coordinates {($\bm k=2$, 22) ($\bm k=3$, 30) ($\bm k=4$, 5) ($\bm k=5$, 4)}; %($\bm k=6$, 7)};
\legend{$\mathcal F_{\operatorname{nf}}$, $\mathcal F_{\operatorname{dc}}$, $\mathcal F_{\operatorname{soc}}$}
\end{axis}
\end{tikzpicture}
\caption{Number of iterations of the cutting-plane algorithm applied to formulations $\mathcal F_{\operatorname{soc}}$, $\mathcal F_{\operatorname{dc}}$, and $\mathcal F_{\operatorname{nf}}$ in the IEEE single-area RTS96 system with 24 buses.}
\label{fig:iterations_cp_24}
\end{figure}

\subsection{Performance of the cutting-plane algorithm on categories (C2) and (C3)} \label{subsec:c2_c3}
In this section, we present the computational results for the test cases in categories (C2) and (C3). Unlike the instances in category (C1), computing the optimal solution using complete enumeration is time consuming for the instances in categories (C2) and (C3), so we present the results of the cutting-plane algorithm applied to formulations $\mathcal F_{\operatorname{soc}}$, $\mathcal F_{\operatorname{dc}}$, and $\mathcal F_{\operatorname{nf}}$ for instances in category (C1) only. A tolerance of $\varepsilon = 1\%$ is set for every run of the cutting-plane algorithm. A computation time limit of $24$ hours was imposed on all the runs of the algorithm. Table \ref{tab:c2} summarizes the computational behavior of the cutting-plane algorithm and the heuristic for the instances in category (C2) and (C3), respectively. The column headings used in the table are defined as follows:

\noindent \textbf{\# buses}: number of buses (nodes) in the system.\\
\noindent $\bm k$: number of interdicted lines (edges).\\
\noindent $\mathcal F_i$: cutting-plane algorithm applied to formulation $\mathcal F_i$, $i\in \{\operatorname{nf}, \operatorname{dc}, \operatorname{soc}\}$.\\
\noindent \textbf{time}: time in seconds for the cutting-plane algorithm to converge to $\varepsilon$-optimal solution; if the algorithm timed out at $24$ hours, then ``TO'' (timed-out) status is reported. \\
\noindent \textbf{iter}: number of iterations of the cutting-plane algorithm.\\
\noindent \textbf{sol}: $\varepsilon$-optimal objective value (in MW) if the cutting-plane algorithm  terminated within 24 hours or best feasible solution obtained at the end of 24 hours (if algorithm timed out, the gap in \% is reported instead of the objective value; the gaps are enclosed in parentheses to differentiate them from the $\varepsilon$-optimal solutions). We note that this column reports the cost of the heuristic solution, obtained using the cutting-plane algorithm, for formulations $\mathcal F_{\operatorname{soc}}$ and $\mathcal F_{\operatorname{dc}}$.\\
\noindent \textbf{ac}: the objective value of the feasible solution to the PNK obtained via primal recovery. This objective value is not reported for instances that timed out. \\

We observe from Table \ref{tab:c2} that the SOC relaxation of the PNK given by formulation $\mathcal F_{\operatorname{soc}}$ performs the best from a computational perspective. The cutting-plane algorithm consistently takes the least number of iterations when applied to the SOC relaxation. The cutting-plane algorithm performs the worst when applied to the NF approximation of the PNK. This can be attributed to the fact that the NF approximation has a larger feasible region of the PNK, whereas the SOC relaxation is a smaller representation of the feasible set of the PNK. Despite the fact that the cut in \eqref{eq:woodscut} is not necessarily valid for the SOC relaxation and the DC approximation, the primal recovery procedure to obtain a feasible solution to the MINLP formulation for the PNK is observed to provide similar objective values to both these formulations. Hence, it is hard to distinguish between or identify any trends in the actual solutions produced by the cutting-plane algorithm when applied to formulations $\mathcal F_{\operatorname{dc}}$ and $\mathcal F_{\operatorname{soc}}$. Fig. \ref{fig:hamming} shows the Hamming distance between the solutions obtained by the cutting-plane algorithm on the three formulations for the WECC 240 test system; it can be observed from the figure that the Hamming distance between the interdiction plans obtained using formulations $\mathcal F_{\operatorname{soc}}$ and $\mathcal F_{\operatorname{dc}}$ is zero for $\bm k \in \{2,3,4,5\}$. 

There are outliers to the fact that the cutting-plane algorithm on $\mathcal F_{\operatorname{dc}}$ and $\mathcal F_{\operatorname{soc}}$ results in the same choice of interdiction plans; for instance consider the instance IEEE three-area RTS96 with 73 buses: when $\bm k = 3$, the Hamming distance between the interdiction plans obtained from the two formulations is 2 and the difference between the objective values obtained using the primal recovery procedure is $15\%$. In fact, there are many other problems in power systems \cite{Coffrin2014, Nagarajan2017} where the SOC relaxation and the DC approximation of the power flow constraints produce contrasting results. We delegate to future work the need for further study to determine if the sameness of the solutions between the two formulations holds more generally for the PNK. 

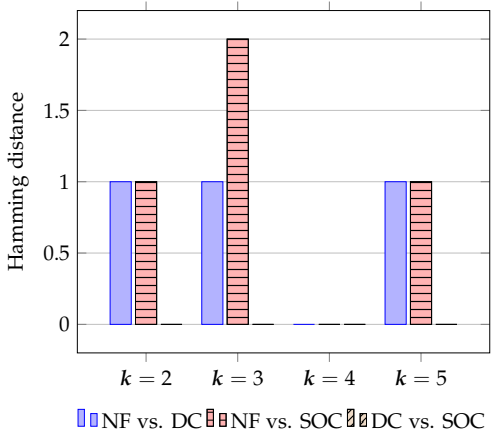
\begin{figure}
\centering
    \begin{tikzpicture}[scale=0.8]
\begin{axis}[
    ybar,
    enlarge x limits=0.25,
    legend style={at={(0.5,-0.15)},draw=none,fill=none,
      anchor=north,legend columns=-1},
    ylabel={Hamming distance},
    symbolic x coords={$\bm k=2$, $\bm k=3$,$\bm k=4$,$\bm k=5$},
    xtick=data,
    ymajorgrids,
    ]
\addplot coordinates {($\bm k=2$,1) ($\bm k=3$,1) ($\bm k=4$,0) ($\bm k=5$,1)};
\addplot[fill=red!30, postaction={pattern=horizontal lines}] coordinates {($\bm k=2$,1) ($\bm k=3$,2) ($\bm k=4$,0) ($\bm k=5$,1)};
\addplot[fill=brown!30, postaction={pattern=north east lines}] coordinates {($\bm k=2$,0) ($\bm k=3$,0) ($\bm k=4$,0) ($\bm k=5$,0)};
\legend{NF vs. DC, NF vs. SOC, DC vs. SOC}
\end{axis}
\end{tikzpicture}
\caption{Hamming distance between the interdiction plans obtained by the cutting-plane algorithm on the three formulations for the WECC 240 test system.}
\label{fig:hamming}
\end{figure}

However, taking into account the computation time reveals the SOC relaxation in formulation $\mathcal F_{\operatorname{soc}}$ to be the clear winner between the three formulations, from both a solution quality and computation time standpoint, for the PNK. A likely reason for the relative computational efficiency of the cutting-plane algorithm applied on $\mathcal F_{\operatorname{soc}}$ is that the nonlinear line losses captured by the SOC relaxation has the property of removing symmetric solutions. In contrast, the DC approximation can potentially allow solutions ($N$-$k$ failures) which produce the same amount of load shed. The number of iterations of the cutting-plane algorithm on $\mathcal F_{\operatorname{soc}}$ and $\mathcal F_{\operatorname{dc}}$ is also consistent with this hypothesis. This type of behaviour was also observed in the literature on similar problems such as Optimal Transmission Switching (see \cite{Hijazi2017,Coffrin2014}). However, it is uncommon. Hence, we will use the SOC relaxation of the PNK for the remaining computational experiments.

\begin{table*}
    \centering
    \small
    \begin{tabular}{ccrrrrrrrrrrrr}
         \toprule
         \# buses & $\bm k$ & \multicolumn{4}{c}{$\mathcal F_{\operatorname{nf}}$}  & \multicolumn{4}{c}{$\mathcal F_{\operatorname{dc}}$}  & \multicolumn{4}{c}{$\mathcal F_{\operatorname{soc}}$}  \\
         \cmidrule(lr){3-6}\cmidrule(lr){7-10}\cmidrule(lr){11-14}
         & & time & iter & sol & ac & time & iter & sol & ac & time & iter & sol & ac \\
         \midrule
         73 & 2 & 235.06 & 242 & 28.75 & 28.75 & 31.09 & 107 & 28.75 & 28.75 & 29.61 & 117 & 28.75 & 28.75 \\
         73 & 3 & 45804.10 & 1256 & 15.15 & 15.15 & 580.29 & 382 & 15.52 & 15.52 & 66.27 & 144 & 19.19 & 19.18 \\
         73 & 4 & 1715.02 & 336 & 20.48 & 20.96 & 16.66 & 47 & 20.48 & 20.96 & 6.30 & 22 & 20.96 & 20.96 \\
         73 & 5 & 18438.44 & 641 & 11.06 & 11.32 & 61.12 & 99 & 11.06 & 11.32 & 9.69 & 35 & 11.32 & 11.32 \\
         73 & 6 & TO & 911 & (26.97) & ----- & 118.49 & 127 & 5.97 & 6.11 & 19.18 & 60 & 6.11 & 6.11\\
         73 & 7 & TO & 794 & (31.19) & ----- & 335.76 & 167 & 3.22 & 3.30 & 29.87 & 77 & 3.30 & 3.30 \\
         73 & 8 & TO & 716 & (26.17) & ----- & 511.23 & 202 & 1.71 & 1.73 & 40.43 & 86 & 1.74 & 1.74 \\
         73 & 9 & TO & 687 & (23.68) & ----- & 640.90 & 194 & 0.91 & 0.91 & 47.84 & 93 & 0.91 & 0.91 \\
         73 & 10 & TO & 673 & (22.23) & ----- & 435.06 & 217 & 0.47 & 0.47 & 57.70 & 111 & 0.47 & 0.47 \\
        \addlinespace
         118 & 2 & 208.81 & 260 & 21.22 & 21.11 & 9.90 & 33 & 21.12 & 21.11 & 8.72 & 28 & 21.12 & 21.11 \\
         118 & 3 & 13565.87 & 1199 & 11.40 & 11.40 & 24.83 & 83  & 11.40 & 11.40 & 26.87 & 75  & 11.40 & 11.40 \\
         118 & 4 & TO & 1809 & (16.99) & ----- & 61.30 & 124 & 6.83 & 6.83 & 43.81  & 101 & 6.92 & 8.52 \\
         118 & 5 & TO & 1497 & (16.61) & ----- & 348.21 & 245 & 3.67 & 4.57 & 115.36 & 167 & 3.73 & 4.60 \\
         118 & 6 & TO & 1293 & (26.97) & ----- & 1205.63 & 348 & 2.02 & 2.02 & 319.29 & 252 & 2.02 & 2.02 \\
         118 & 7 & TO & 1095 & (31.19) & ----- & 3161.65 & 458 & 1.09 & 1.09 & 508.70 & 297 & 1.09 & 1.09 \\
         118 & 8 & TO & 951 & (26.17) & ----- & 8125.70 & 588 & 0.58 & 0.58 & 1543.57 & 447 & 0.58 & 0.58 \\
         118 & 9 & TO & 866 & (23.68) & ----- & 15615.67 & 677 & 0.30 & 0.30 & 1777.69 & 444 & 0.31 & 0.36 \\
         118 & 10 & TO & 840 & (22.23) & ----- & 20622.50 & 708 & 0.16 & 0.16 & 2867.04 & 507 & 0.16 & 0.19 \\
         \addlinespace
         240	&	2	&	9.60	&	21	&	2227.27	&	2287.13	&	1.66	&	6	&	2331.26	&	2482.68	&	7.35	&	8	&	2393.18	&	2482.68	\\
240	&	3	&	14.13	&	41	&	1329.31	&	1418.38	&	2.73	&	13	&	1381.54	&	1471.39	&	7.70	&	8	&	1422.93	&	1471.39	\\
240	&	4	&	15.71	&	46	&	778.54	&	831.09	&	2.87	&	13	&	796.09	&	831.09	&	10.17	&	12	&	808.91	&	831.09	\\
240	&	5	&	39.54	&	88	&	430.66	&	451.17	&	3.09	&	15	&	437.94	&	458.66	&	11.04	&	13	&	447.24	&	458.66	\\
240	&	6	&	127.70	&	169	&	235.24	&	242.83	&	4.65	&	22	&	236.49	&	247.68	&	18.16	&	23	&	241.77	&	248.79	\\
240	&	7	&	526.72	&	312	&	127.03	&	131.13	&	4.88	&	23	&	129.09	&	134.86	&	19.83	&	25	&	130.55	&	134.35	\\
240	&	8	&	1293.63	&	409	&	69.04	&	71.04	&	8.94	&	36	&	70.05	&	72.19	&	20.50	&	27	&	70.82	&	72.83	\\
240	&	9	&	2362.78	&	473	&	37.28	&	38.36	&	11.08	&	44	&	37.88	&	38.98	&	20.46	&	27	&	38.32	&	39.39	\\
240	&	10	&	3939.01	&	523	&	20.09	&	20.70	&	8.02	&	43	&	20.30	&	21.17	&	17.62	&	23	&	20.65	&	21.24	\\
         \bottomrule
    \end{tabular}
    \caption{Computational results for instances in categories (C2) and (C3).}
    \label{tab:c2}
\end{table*}

\subsection{Effect of the failure probability distributions} \label{subsec:prob}
For all the computational experiments thus far we either used the line-failure probabilities from publicly available data (reliability data) or generated data from a uniform distribution in the range defined by the reliability data. Now, we study the effect of a change in distribution of line-failure probabilities on the solution obtained by the cutting-plane algorithm. For the purposes of this study, we restrict our attention to the SOC relaxation of the PNK given by the formulation $\mathcal F_{\operatorname{soc}}$ and the instance WECC 240 from category (C3). 

For ease of exposition, we will denote the line-failure probability of line $(i,j) \in \mathcal E$ generated using the reliability data by $\bm{Pr}_{ij}^r$ \cite{Grigg1999}. We perform two sets of computational experiments, referred to as (D1) and (D2). The first set of experiments is aimed at examining how an increase in line-failure probability in a certain geographical area affects the objective value obtained using the cutting-plane algorithm. This experiment is intended to simulate a severe event that could potentially increase the line-failure probabilities in a geographical region. 
\begin{figure}[h!]
\centering
\includegraphics[scale=0.2,cfbox=black 1pt 1pt]{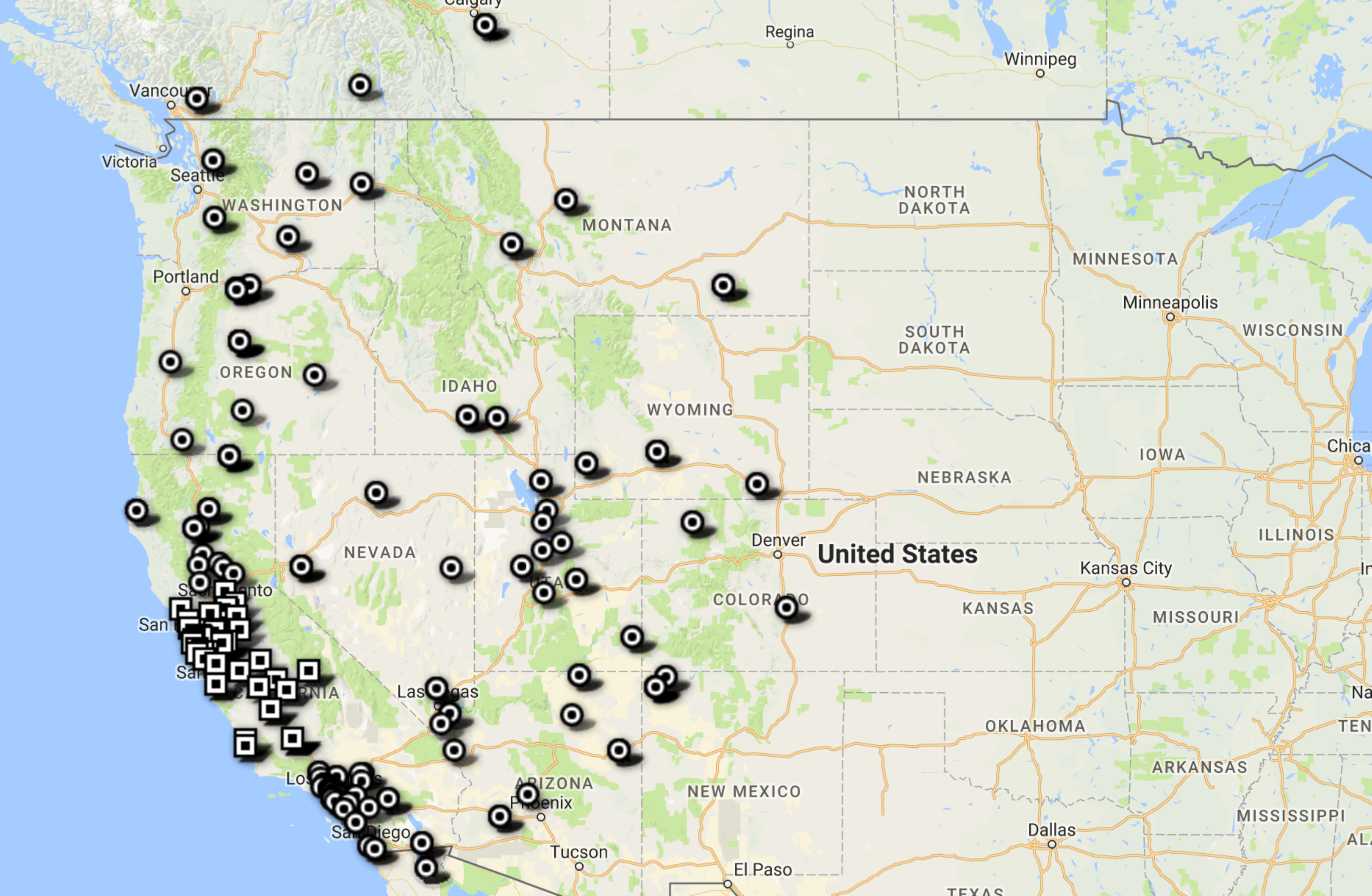}
\caption{Buses that are assumed to be affected by a severe event on the coast of California are represented using squares and the other buses in the WECC 240 test system are represented using circles.}
\label{fig:hurricane}
\end{figure}
To that end, we first choose a set of $117$ transmission lines, $\mathcal L \subset \mathcal E$, in the state of California using the geolocated WECC 240 test instance. The set of buses on which these lines are incident is shown in Fig. \ref{fig:hurricane}; this subset of buses is assumed to be affected by a severe event on the coast of California. The line-failure probability for each line in $\mathcal E$ is generated using the following equation: 
\begin{flalign}
\bm{Pr}_{ij} = \begin{cases} \bm{Pr}_{ij}^r + n \cdot \frac{1-\bm{Pr}_{ij}^r}{5} & \mbox{if } (i,j) \in \mathcal L, \\ \bm{Pr}_{ij}^r & \mbox{otherwise}. \end{cases} \label{eq:d1}
\end{flalign}
We use $n=\{1,2,3\}$ to generate three sets of line-failure probabilities. We also note that the higher the value of $n$, the higher the probability of line failure for the lines in set $\mathcal L$. Hence, $n=3$ can correspond to a more severe event than $n \in \{1,2\}$. Also, when $n=0$, the line-failure probabilities reduce to the probabilities that were obtained using the reliability data, studied in the previous section.

\begin{figure}[h!]
\centering
\begin{tikzpicture}[scale=0.8]
    \begin{axis}[xlabel=$\bm k$, ylabel=\text{obj. value (MW)}, legend style={draw=none},grid=major,]
    \addplot[only marks,color=brown,mark=*,mark size=2] plot coordinates {
        (2,2393.18) (3,1422.93) (4,808.91) (5,447.24) (6,241.77) (7,130.55) (8,70.82) (9,38.32) (10,20.65) (11,11.17) (12,5.98) (13,3.23) (14,1.71) (15,0.91)
    };
    \addplot[only marks,color=red,mark=diamond*,mark size=2] plot coordinates {
        (2,2551.52) (3,1592.15) (4,1013.70) (5,632.55) (6,394.71) (7,246.30) (8,151.72) (9,93.46) (10,56.82) (11,34.55) (12,20.44) (13,12.27) (14,7.19) (15,4.30)
    };
\addplot[only marks,color=blue,mark=square*,mark size=2] plot coordinates {
        (2,3429.83) (3,2569.06) (4,1844.77) (5,1340.52) (6,962.59) (7,685.30) (8,496.24) (9,353.32) (10,249.44) (11,177.60) (12,125.39) (13,87.02) (14,60.39) (15,41.62)
    };
\addplot[only marks,color=black,mark=triangle*,mark size=2] plot coordinates {
        (2,4411.60) (3,3790.12) (4,3107.48) (5,2565.72) (6,2267.55) (7,1931.06) (8,1599.71) (9,1292.57) (10,1055.88) (11,857.43) (12,696.81) (13,562.48) (14,457.17) (15,369.21)
    };
    \legend{$n=0$,$n=1$,$n=2$,$n=3$}
    \end{axis}
\end{tikzpicture}
\caption{Plot of the objective value vs. $\bm k$ using the probability data using Eq. \eqref{eq:d1}. $n=0$ shows the plot for the case when line-failure probabilities are obtained using the reliability data, $\bm{Pr}_{ij}^r$.}
\label{fig:case1}
\end{figure}
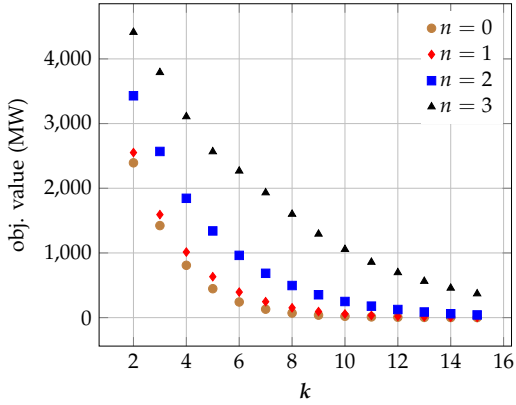

Fig. \ref{fig:case1} shows the objective value obtained by the cutting-plane algorithm for formulation $\mathcal F_{\operatorname{soc}}$ for the WECC 240 test system for $\bm k \in \{2, \dots,15\}$ using the line-failure probabilities generated for the computational experiments in set (D1). It is clear from the figure that the objective value increases with increasing line-failure probabilities; this is expected given the objective function of the PNK. 

The second set of experiments, (D2), is aimed at examining the effect of the line-failure probability distribution on the solution obtained by the cutting-plane algorithm. To make the comparison fair, we assume a probability budget, which is defined as the sum of the line-failure probabilities of all the lines using the probabilities generated using the reliability data, $\bm{Pr}_{ij}^r$. The probability budget is defined as $\mathcal B := \sum_{(i,j) \in \mathcal E} \bm{Pr}_{ij}^r$. We will refer to the line-failure probabilities $\bm{Pr}_{ij}^r$ as ``$\operatorname{rel}$'', where $\operatorname{rel}$ is an abbreviation of reliability. We consider three cases for the line-failure probabilities: (1) we assume each $\bm{Pr}_{ij}$ is exactly $0.5$ (this is equivalent to the deterministic variant of the PNK and we will refer to this case as ``$\operatorname{det}$''), (2) we assume $\bm{Pr}_{ij} \sim U(0,1)$ (a uniform distribution in the interval $(0,1)$), and (3) we assume $\bm{Pr}_{ij} \sim \operatorname{texp}(1)$, where ``$\operatorname{texp}(1)$'' denotes a truncated exponential distribution in $(0,1)$ with rate $1$. Once we generate the line-failure probabilities for all three cases, these values are normalized to sum to the probability budget $\mathcal B$.

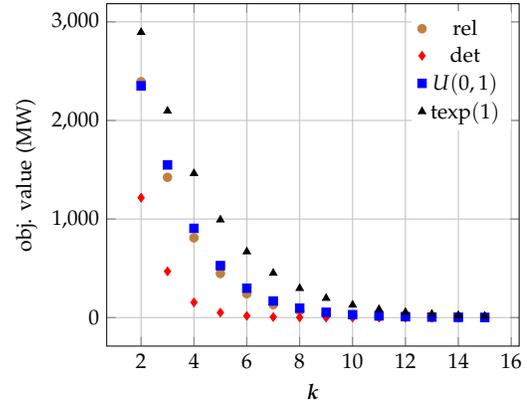
\begin{figure}[h!]
\centering
\begin{tikzpicture}[scale=0.8]
    \begin{axis}[
        xlabel=$\bm k$,
        ylabel=\text{obj. value (MW)},
        grid=major,
        legend style={draw=none},
    ]
    \addplot[only marks,color=brown,mark=*,mark size=2] plot coordinates {
        (2,2393.18) (3,1422.93) (4,808.91) (5,447.24) (6,241.77) (7,130.55) (8,70.82) (9,38.32) (10,20.65)
        (11,11.17) (12,5.98) (13,3.23) (14,1.71) (15,0.91)
    };
    \addplot[only marks,color=red,mark=diamond*,mark size=2] plot coordinates {
        (2,1216.08) (3,468.76) (4,153.88) (5,50.46) (6,16.95) (7,5.39) (8,1.69) (9,0.52) (10,0.16) (11,0.05) (12,0.01) (13,0.00) (14,0.00) (15,0.00)    
    };
    \addplot[only marks,color=blue,mark=square*,mark size=2] plot coordinates {
        (2,2351.22) (3,1548.61) (4,905.36) (5,526.82) (6,298.18) (7,168.18) (8,95.05) (9,53.07) (10,29.51) (11,16.49) (12,9.15) (13,5.05) (14,2.80) (15,1.54)    
    };
    \addplot[only marks,color=black,mark=triangle*,mark size=2] plot coordinates {
        (2,2894.73) (3,2096.48) (4,1462.96) (5,991.89) (6,668.55) (7,451.03) (8,296.35) (9,196.20) (10,127.92) (11,83.15) (12,54.18) (13,35.22) (14,22.79) (15,14.66)    
    };
    \legend{$\operatorname{rel}$\\$\operatorname{det}$\\$U(0,1)$\\$\operatorname{texp}(1)$\\}
\end{axis}
\end{tikzpicture}
\caption{Plot of the objective value vs. $\bm k$ for different distributions of line-failure probability.}
\label{fig:case2}
\end{figure}

Fig. \ref{fig:case2} shows the objective value obtained by the cutting-plane algorithm decreasing for increasing values of $\bm k$; this is because the larger the value of $\bm k$, the lower the probability of an $N$-$k$ failure (this observation is also valid for Fig. \ref{fig:case1}). The sum of the individual line-failure probabilities is normalized to a constant value for all four cases; this can be interpreted as moving the line-failure probabilities from one line to another. This change does not seem to affect the value of $\bm k$ at which the objective becomes comparatively small. This leads to a hypothesis that for a given system and for a constant value of probability budget $\mathcal B$, there exists a value of $\bm k$ greater than which the likelihood of the system sustaining damage is low. Also, the objective value for a fixed $\mathcal B$ is observed to be greatest when $\bm{Pr}_{ij} \sim \operatorname{texp}(1)$ and least for the ``$\operatorname{det}$'' case, when all the probabilities take a value of $0.5$ before normalization. The fact that this objective value is lowest in the ``det'' case and increases significantly in the ``exp'' case suggests that the interdiction plan obtained by solving the PNK is sensitive to  the nature of the failure distribution. This observation is also corroborated by Fig. \ref{fig:regions}.
\begin{figure}[h!]
\centering
\includegraphics[scale=0.4, cfbox=black 1pt 1pt]{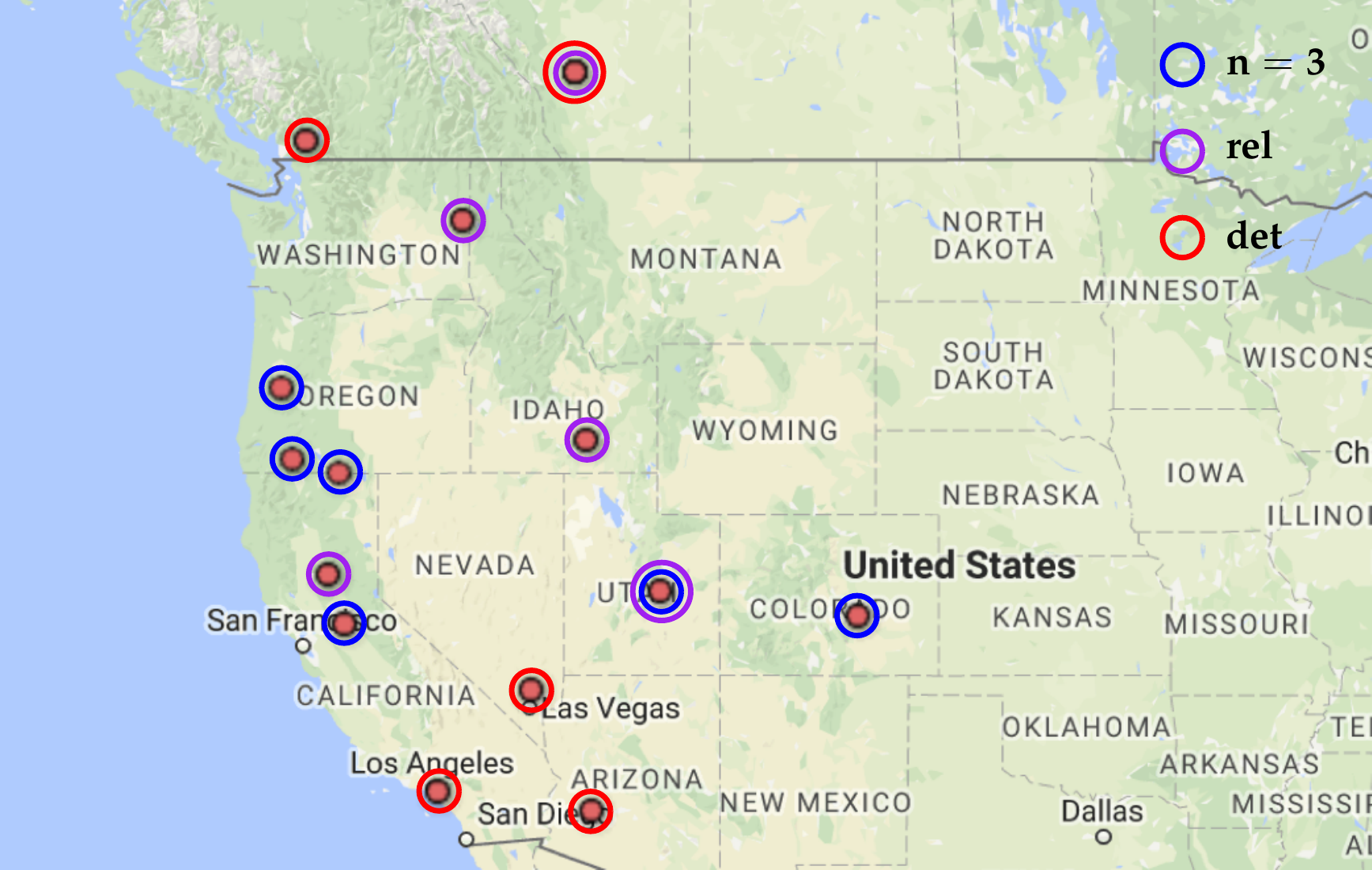}
\caption{Regions that are affected for the different experiments on the WECC 240 instance.}
\label{fig:regions}
\end{figure}
Fig. \ref{fig:regions} shows the regions in which the lines are interdicted for the case when a severe event occurs on the coast of California, with line-failure probabilities generated using Eq. \eqref{eq:d1} with the value of $n$ set to $3$, the ``det'' case, and the ``rel'' case. 

\subsection{Scalability of the cutting-plane algorithm} \label{subsec:scalability}
In this section, we demonstrate the scalability of the cutting-plane algorithm on the very large test instances from category (C4). Given the computational and solution quality superiority of $\mathcal F_{\operatorname{soc}}$ on instances in categories (C2) and (C3),
for the purposes of this study, we restrict our attention to $\mathcal F_{\operatorname{soc}}$. Furthermore, we increase the tolerance from $\varepsilon = 1\%$ to $\varepsilon = 5\%$ and set a computation time limit of $24$ hours for every run of the cutting-plane algorithm. An $\varepsilon$ value of $5\%$ was chosen for this set of experiments as 
this value of $\varepsilon$ provided a reasonable trade-off between computation time and the solution quality. Table \ref{tab:c4} summarizes the computational behavior of the cutting-plane algorithm applied to the formulation $\mathcal F_{\operatorname{soc}}$ for the ``PEGASE'' and ``Polish'' test instances. The additional column headings, used in the Table \ref{tab:c4} are defined as follows:

\noindent \textbf{feas}: objective value (MW) of the heuristic solution provided by the cutting-plane algorithm for $\varepsilon=5\%$. \\
\noindent \textbf{gap}: relative gap in $\%$ \emph{i.e.}, $\frac{(\bar f - f^*)}{f^*}\cdot 100 \%$ after the termination of the cutting-plane algorithm.

The results in Table \ref{tab:c4} indicate that the cutting-plane algorithm can scale to very large test instances of the order of $2000$ buses. Even for the Polish instances that timed out in the $\bm k = 6,\dots,10$ cases, the relative gap is below $6\%$ further corroborating the effectiveness of the algorithm.

\begin{table*}
    \centering
    \small
    \begin{tabular}{crrrrrrrrrr}
         \toprule
         $\bm k$ & \multicolumn{5}{c}{PEGASE 1354-bus test system}  & \multicolumn{5}{c}{Polish 2383-bus test system}  \\
         \cmidrule(lr){2-6}\cmidrule(lr){7-11}
          & time & iter & feas & gap & ac & time & iter & feas & gap & ac \\
         \midrule
        2 & 133.25 & 57 & 255.22 & 0.12 & 255.23 & 1046.78 & 64 & 100.72 & 0.83 & 109.77 \\
    3 & 159.27 & 66 & 172.07 & 0.53 & 173.74 & 2937.39 & 164 & 58.88 & 1.31 & 63.02 \\
    4 & 398.36 & 125 & 102.94 & 2.10 & 103.82 & 7520.70 & 335 & 34.75 & 0.64 & 36.85 \\
    5 & 954.42 & 215 & 60.27 & 0.02 & 60.73 & 17126.15 & 557 & 18.76 & 5.00 & 19.90 \\
    6 & 2861.25 & 352 & 32.55 & 5.00 & 32.79 & TO & 1106 & 9.95 & 5.98 & 10.74 \\
    7 & 5469.59 & 425 & 18.19 & 5.00 & 19.46 & TO & 1015 & 5.87 & 5.02 & 7.23 \\
    8 & 28372.89 & 721 & 10.32 & 2.59 & 10.57 & TO & 913 & 3.10 & 6.75 & 3.82 \\
    9 & 28401.17 & 624 & 5.43 & 4.96 & 5.75 & TO & 832 & 1.70 & 6.77 & 2.08 \\
    10 & 29980.16 & 545 & 3.09 & 4.06 & 3.44 & TO & 750 & 0.89 & 7.39 & 1.10 \\

         \bottomrule
    \end{tabular}
    \caption{Computational results for instances in category (C4).}
    \label{tab:c4}
\end{table*}

\section{Conclusion} \label{sec:conclusion}
In this paper, we have presented an MINLP formulation, a convex (SOC) relaxation, and linear approximations to the probabilistic interdiction problem in power transmission networks. To the best of our knowledge, this is the first study of the $N$-$k$ problem on power transmission systems that incorporates the convex relaxation of the AC power flow constraints. In summary, we have shown the following: (1) the computational effectiveness of the generic cutting-plane algorithm based on work in the literature for the deterministic variant of the PNK with the DC approximation \cite{Salmeron2009} extends to the probabilistic generalization and also to convex relaxations of the AC power flow equations; (2) the nonlinear convex relaxation (SOC) of the AC power flow equations accelerates the convergence of the cutting-plane algorithm in comparison to linear DC/NF approximations; (3) the interdiction plan obtained using the cutting-plane algorithm is sensitive to changes in the component-failure probabilities, and this sensitivity can be used to compute interdiction plans in smaller subnetworks of the original system by spatially biasing the probabilities in the regions of interest; and (4) the PNK is very effective at identifying high-risk failure plans in the case of natural or manmade severe events. Future work should focus on modeling the dependence on component-failure probabilities, dependence on severe event probabilities, as well an in depth comparison of solutions obtained using a deterministic $N$-$k$ model. Further work should also focus on stochastic variants of the $N$-$k$ interdiction problem on power systems and comparing the solutions with the PNK. Finally, including models of cascading failures in the context on the $N$-$k$ problems remains a computational challenge and is also an important avenue of future work.

\section*{Acknowledgements} 
This work was supported by the U.S. Department of Energy through the LANL/LDRD program and Center for Nonlinear Studies.

\bibliographystyle{unsrt}
\bibliography{Nkref.bib}
% \printbibliography

\end{document}